\documentclass{jfm}

\usepackage{amsmath,amssymb}
\usepackage{bm}
\usepackage{graphicx}
\usepackage{newtxtext}
\usepackage{newtxmath}
\usepackage{natbib}
\usepackage{hyperref}

\hypersetup{
    colorlinks = true,
    urlcolor   = blue,
    citecolor  = blue,
}

\newtheorem{theorem}{Theorem}

\newtheorem{lemma}[theorem]{Lemma}
\newtheorem{corollary}[theorem]{Corollary}
\newtheorem{proposition}[theorem]{Proposition}
\newtheorem{definition}[theorem]{Definition}
\newtheorem{remark}[theorem]{Remark}

\newtheorem{conjecture}[theorem]{Conjecture}

\newcommand{\bu}{\boldsymbol{u}}
\newcommand{\bv}{\boldsymbol{v}}
\newcommand{\bU}{\boldsymbol{U}}
\newcommand{\bw}{\boldsymbol{w}}
\newcommand{\bx}{\boldsymbol{x}}
\newcommand{\bk}{\boldsymbol{k}}
\newcommand{\bq}{\boldsymbol{q}}
\newcommand{\bn}{\hat{\boldsymbol{n}}}
\newcommand{\T}{^{\mathrm{T}}}
\newcommand{\dd}{\,\mathrm{d}}
\newcommand{\e}{\mathrm{e}}
\newcommand{\im}{\mathrm{i}}
\newcommand{\Qstar}{Q_{*}}
\newcommand{\Gf}{G_{\!f}}

\title{A transport geometry of acoustic analogies:\\
exact holonomy of source re-attribution\\ and its observable consequences}

\author{
Sparsh Sharma\aff{1}
  \corresp{\email{sparsh.sharma@dlr.de}},
}

\affiliation{\aff{1}German Aerospace Center (DLR), 38108 Braunschweig, Germany}

\begin{document}

\maketitle

\begin{abstract}
The source term of an acoustic analogy is not unique: rearrangements of
the Navier--Stokes equations about different base flows, or in different
dependent variables, attribute the same far field to different
``sources''.  This non-uniqueness has been recognised since the earliest
days of the subject but has never been given a quantitative structure.
We supply one.  Acoustic analogies are organised into a fibre-like family
over the space of effective media, and the operation of passing between
analogies---re-gauging---is formalised as a transport by the unique
frequency-preserving, rotation-free linear space--time maps that reduce
each convected wave operator to the d'Alembertian.  Three exact results
follow.  First, the classical convected analogies are not closed under
this transport: two successive uniform-flow re-gaugings land on a moving
medium with an anisotropic sound-speed tensor, forcing precisely the
generalised (anisotropic) base media introduced by Goldstein
(\emph{J.~Fluid Mech.}~488, 2003).  Second, the discrete holonomy of
transport is computed in closed form: it consists of a frequency dilation
$\rho$, given exactly by
$\rho^{-2}=[(c^{2}-\bu_{1}\!\cdot\!\bu_{2})^{2}-c^{2}|\bu_{1}|^{2}]
/[c^{2}(c^{2}-|\bu_{1}|^{2})]$, and a rotation that vanishes to fourth
order in Mach number; transport degenerates on an exact ``sonic horizon''
in analogy space, located for equal collinear steps at
$M=(\sqrt5-1)/2$.  Third, in the continuum limit the boost sector is
flat, and all curvature is localised in the anisotropy directions of the
effective medium, where it equals the commutator of sound-speed-tensor
increments.  These results carry an operational consequence.  Every exact
analogy returns the identical true far field, so the geometry describes
no physical change of the radiated sound; what it quantifies is the
discrepancy incurred off shell, when one and the same \emph{modelled}
source spectrum is transported through different analogies, as it
routinely is in hybrid prediction.  For two pipelines that reuse a common
model across analogies related by a loop, the far-field predictions differ
by an exact, parameter-free rigid spectral dilation by $\rho$ (reaching
$36\%$ for two Mach-$0.45$ re-gauging steps) and a rigid directivity
rotation---an inter-method bias, not a change in the physics.  An exact
far-field law for anisotropic generalised media---sampling of the source
spectrum on sonic ellipsoids along group rays---extends the result to open
paths.  All identities are verified symbolically and numerically; the
verification scripts are provided as supplementary material.
\end{abstract}

\noindent\textbf{Key words:} aeroacoustics, acoustic analogy

\section{Introduction}\label{sec:intro}

An acoustic analogy is an exact rearrangement of the Navier--Stokes
equations into the form
\begin{equation}
  L\,\varphi \;=\; s ,
  \label{eq:analogy-schematic}
\end{equation}
in which $L$ is a linear wave-type operator, $\varphi$ is a dependent
variable that reduces to the acoustic pressure (or density) fluctuation
in the far field, and everything else is declared a ``source'' $s$.  The
construction is due to \citet{lighthill1952}: with $L$ the
quiescent-medium wave operator and $\varphi=\rho'$, the whole of
nonlinear gas dynamics is compressed into the quadrupole source
$\partial_i\partial_jT_{ij}$, and dimensional analysis of that source
yields the celebrated eighth-power law
\citep{lighthill1952,lighthill1954}.  Because the rearrangement is an
identity, nothing is lost until the source is modelled; and because the
propagation problem is then exactly solvable, the modelling effort can be
concentrated entirely on the source statistics, for which the theory of
turbulence supplies structure \citep{lighthill1954,ribner1962}.
\citet{crow1970} subsequently placed the underlying compact-source
asymptotics on a systematic singular-perturbation footing;
\citet{curle1955} and \citet{fwh1969} extended the identity to flows
containing stationary and moving surfaces; and the framework in this form
became, and remains, the workhorse of aerodynamic noise prediction
\citep{crighton1975,goldstein1976}.

The freedom that Lighthill exploited---declare part of the dynamics
`propagation' and call the remainder `source'---can, however, be
exercised in infinitely many ways, and the subsequent literature
exercised it thoroughly.  One line of development moved mean-flow effects
out of the source and into the operator: building on the shear-flow
propagation analysis of \citet{pridmorebrown1958}, the analogies of
\citet{phillips1960} and \citet{lilley1974} transferred convection and
refraction to variable-coefficient wave operators, at the price of source
terms whose decomposition proved delicate
\citep{goldstein2001,musafir2007}; the propagation theory for moving
streams is developed at length by \citet{dowling1978}, and in the
textbook accounts of \citet{goldstein1976} and \citet{morseingard1968}.
A second line reorganised the source around vorticity, on the physical
intuition that vortical motion is the engine of sound generation: the
vortex-sound formulations of \citet{powell1964}, \citet{howe1975} and
\citet{mohring1978}, and the dilatation form of \citet{ribner1962}.  A
third changed the dependent variables altogether, notably to the momentum
potential and fluctuating total enthalpy of \citet{doak1989,doak1998}.
A fourth, driven by computation, introduced flow splittings designed to
produce well-behaved source terms for hybrid methods
\citep{hardin1994,ewert2003}, statistical source models for
Reynolds-averaged prediction \citep{tamauriault1999}, and source fields
educed directly from unsteady simulation
\citep{colonius1997,freund2001,bogey2002}; this practice is surveyed by
\citet{coloniuslele2004} and \citet{wangfreundlele2006}.  The development
culminates in the generalised acoustic analogy of \citet{goldstein2003}:
the equations of motion can be rearranged as linear perturbation
equations about an essentially \emph{arbitrary} base flow, each choice
producing its own operator, its own dependent variables and its own
residual stresses, and the resulting family has been applied
quantitatively to jet noise with base-flow statistics taken from
modelling or from large-eddy simulation
\citep{goldsteinleib2008,karabasov2010}.  Every one of these
constructions is exact, and every one attributes the same radiated field
to a different source.

That the attribution is not unique was recognised at once and has been
debated ever since.  \citet{fwh1969} and \citet{mohring1978} exhibited
distinct equivalent source distributions producing identical far fields.
\citet{ffowcswilliams1982} made the point programmatically, in a paper
whose title---`fact and fiction'---set the register of the debate, and
\citet{goldstein1984} surveyed the competing formulations.
\citet{tam1998} identified the built-in non-uniqueness as a fundamental
obstacle to declaring any source description `correct', and argued that
distinct components of jet noise may demand distinct descriptions, a
two-source position for which \citet{tam2008} later assembled a body of
experimental evidence.  \citet{goldstein2005} drew the structural
conclusion: since the source of the generalised analogy depends on an
arbitrary choice of base flow, the true sources of aerodynamic sound are
unlikely to be identifiable in any realistic turbulent motion.  The
dependence has meanwhile been exhibited concretely in computation:
\citet{samanta2006} drove several analogies with identical
direct-simulation data and obtained identical far fields, which then
degraded in analogy-dependent ways once errors were introduced into the
inputs; and \citet{sinayoko2011} constructed silent base flows and the
corresponding source fields, demonstrating that the source depends on
both the base flow and the choice of dependent variables.  The modern
wave-packet literature \citep{jordancolonius2013,cavalieri2019} partially
sidesteps the debate by modelling coherent structures educed from data
rather than analogy-defined sources; yet even there a propagation
operator must be chosen, and the model must be transported into it.

A parallel and largely separate literature approached the same ambiguity
from the observability side.  For a \emph{fixed} operator, the sources
producing a given exterior field form an affine family differing by
non-radiating distributions; the structure of such distributions is
classical
\citep{devaneywolf1973,bleistein1977,porterdevaney1982,marengo2000},
entered acoustics with \citet{jessel1973} and \citet{kempton1976}---who
observed that the ambiguity is precisely what active noise control
exploits---and was developed in the aeroacoustic setting by
\citet{doak1988} and, comprehensively, by \citet{musafir2013}.
Near-field acoustic holography realises the practical counterpart: only
the radiating part of a source is reconstructible from exterior data
\citep{maynard1985}.  The lesson of this strand is that the fibre of
equivalent sources over a \emph{single} analogy is well understood; what
has remained unformalised is the relation \emph{between} the fibres of
different analogies.

A third strand supplies the raw material for such a formalisation.  For
uniformly moving media, the reduction of the convected wave equation to
the classical one is itself classical: the Prandtl--Glauert--Lorentz
transformation, given a definitive treatment by \citet{taylor1978} with
wind-tunnel/flight-test equivalence in view, and revisited by
\citet{chapman2000}, who derived the associated similarity variables and
emphasised that the reduction is of Lorentz type---it does not correspond
to a Galilean change of frame.  That the convected operator carries a
Lorentzian geometry, an acoustic metric with the sound speed in the role
of the light speed, is the founding observation of the analogue-gravity
programme \citep{unruh1981,visser1998,barcelo2011}, and has been
advocated within aeroacoustics through the acoustic space--time of
\citet{gregory2015}.  In all of this work, however, the transformations
are employed as isolated devices: one map solves one problem.  Their
composition properties---the group they generate, whether the family of
operators they act on is closed under composition, and what happens to a
\emph{source model} carried by them from one analogy to another---have not, to our
knowledge, been investigated.  It is exactly in these composition
properties, we shall show, that the quantitative structure of the
source-attribution ambiguity resides.

What is missing from all three strands, then, is not the observation of
non-uniqueness but its quantitative structure: a description of the space
of analogies, of the canonical maps between them, and of the invariants
and obstructions of those maps.  The present paper supplies this
structure at the level at which it can be made exact---constant-%
coefficient effective media---and shows that it already has observable
consequences for the assembly of noise-prediction pipelines.  Two
technical obstacles stand in the way of such a structure, and both are
resolved below.  The first is closure: composed uniform-flow reductions
leave the classical operator family, so the space on which transport
acts must first be identified---it turns out to be forced
(theorem~\ref{thm:nonclosure}).  The second is observability at the
endpoints: the media generated by transport are anisotropic, and in an
anisotropic moving medium the far field is organised along group rays
on a sonic ellipsoid rather than along wavevector directions, so a
transported source model can be confronted with far-field data only
once an exactly equivariant radiation law is available for such media;
we derive that law in closed form (lemma~\ref{lem:anisofar}).
Specifically, we establish the following.
\begin{enumerate}
\item[(i)] At the exact level all analogies are far-field equivalent,
trivially (proposition~\ref{prop:trivial}), and the difference between
any two source functionals sharing a dependent variable is a computable
cocycle, purely kinematic for the Lighthill/convected-Lighthill pair
(lemmas~\ref{lem:cocycle} and~\ref{lem:kinematic}): re-attribution is
bookkeeping, not dynamics.
\item[(ii)] The linear space--time maps by which fixed-frequency source
models can be carried between analogies are characterised
(lemma~\ref{lem:block}); each admissible effective medium possesses a
unique rotation-free reduction to the d'Alembertian
(proposition~\ref{prop:w3}), whose uniform-flow case is the
Prandtl--Glauert--Lorentz map and whose space--time factorisation is
Galilean composed with Lorentz (lemma~\ref{lem:factorisation}).
\item[(iii)] The classical convected family is \emph{not closed} under
composed transport (theorem~\ref{thm:nonclosure}): two uniform
re-gaugings generate a moving medium with an anisotropic sound-speed
tensor, so that the generalised media of \citet{goldstein2003} are
forced as a closure requirement rather than adopted as a modelling
convenience.
\item[(iv)] The discrete holonomy of transport is computed in closed
form: a frequency dilation $\rho$, given exactly by
theorem~\ref{thm:rho}, together with a rotation only at fourth order in
Mach number (proposition~\ref{prop:M4}); the composition law of
re-gauging is a hybrid of Galilean and Einstein velocity addition,
coinciding with the latter exactly for perpendicular steps
(corollary~\ref{cor:einstein}).
\item[(v)] Transport degenerates on an exact sonic horizon in analogy
space, $M_1M_2+M_1=1$ for collinear steps and the inverse golden ratio
for equal steps (theorem~\ref{thm:horizon}).
\item[(vi)] In the continuum limit the boost sector is exactly flat, and
all curvature is localised in the anisotropy directions of the effective
medium, where it equals the commutator of sound-speed-tensor increments
(theorem~\ref{thm:curvature}).
\item[(vii)] Acting on statistical source models, transport loops produce
exactly two observable channels---rigid spectral dilation by $\rho$ and
rigid directivity rotation---separated exactly between boost and
anisotropy loops (theorem~\ref{thm:pipelines},
corollary~\ref{cor:complementarity}); the far field of an anisotropic
generalised medium is obtained in closed form, sampling the source
spectrum on a sonic ellipsoid along group rays
(lemma~\ref{lem:anisofar}), which extends the pipeline law to open paths
(corollary~\ref{cor:openpath}).  Representative magnitudes are large: a
$36\%$ systematic spectral misalignment for two Mach-$0.45$ re-gauging
steps.
\end{enumerate}

The paper follows this sequence: \S\ref{sec:exact} treats the exact
level; \S\ref{sec:freq} constructs the frequency-respecting maps;
\S\ref{sec:transport} composes them and establishes non-closure and
admissibility; \S\ref{sec:holonomy} computes the exact holonomy and the
horizon; \S\ref{sec:curvature} passes to the continuum limit;
\S\ref{sec:stat} derives the observable consequences; and
\S\ref{sec:discussion} returns to the true-source debate and lists the
open problems.  Every identity asserted in the paper has been verified
symbolically or to machine precision, independently of its derivation;
the verification record is summarised in
appendix~\ref{app:verification} and the scripts are supplied as
supplementary material.

Two remarks on scope.  First, all exact computations below concern
spatially uniform (constant-coefficient) effective media; genuinely
sheared base flows, with variable coefficients, are outside the exact
theory and are discussed as the principal open problem.  The uniform
sector is nevertheless where the classical debate lives---Lighthill
versus convected Lighthill is precisely a uniform re-gauging---and it is
rich enough to force the anisotropic generalisation and to carry
nontrivial holonomy.  Second, the word ``gauge'' is used below only in
the neutral sense of a choice within an equivalence class; no field-%
theoretic apparatus is presupposed.  The geometry that appears is the
elementary geometry of matrix groups.

\section{Acoustic analogies as exact identities}\label{sec:exact}

\subsection{Setting and the triviality of exact equivalence}

We work on $\mathbb{R}^{1+d}$, $d\in\{2,3\}$, with coordinates
$\xi=(t,\bx)$, and consider smooth solutions
$q=(\rho,\bu,p,\dots)$ of the compressible Navier--Stokes equations
\begin{equation}
  \partial_t \rho + \partial_j(\rho u_j) = 0, \qquad
  \partial_t(\rho u_i) + \partial_j(\rho u_i u_j) + \partial_i p
   = \partial_j \sigma_{ij},
  \label{eq:NS}
\end{equation}
together with an energy equation and equation of state, in a class with
uniform state $(\rho_0,p_0,\bu\!\to\!\bU_\infty)$ and decaying
fluctuations at spatial infinity;\ $c$ denotes the ambient sound speed
and primes denote fluctuations about the ambient state.

\begin{definition}[acoustic analogy]\label{def:analogy}
An \emph{acoustic analogy} is a triple $A=(L,\varphi,s)$ consisting of a
linear differential operator $L$ that reduces to a d'Alembertian at
spatial infinity, a map $\varphi$ from flow states to a scalar (or
vector) wave variable whose far-field trace equals the normalised
acoustic pressure fluctuation, and a source functional $s$, such that
\begin{equation}
  L\,\varphi(q) \;=\; s(q)
  \qquad\text{identically on solutions $q$ of \eqref{eq:NS}}.
  \label{eq:analogy-def}
\end{equation}
The analogy is \emph{exact} if \eqref{eq:analogy-def} holds without
approximation.
\end{definition}

The first observation is elementary but, we shall argue, is the precise
mathematical reason why the seventy-year debate about ``true sources''
could not be resolved on its own terms.

\begin{proposition}[triviality of exact far-field equivalence]
\label{prop:trivial}
Call two exact analogies \emph{on-shell far-field equivalent} if, for
every solution $q$ of \eqref{eq:NS}, the far-field traces of
$\varphi_1(q)$ and $\varphi_2(q)$ coincide.  Then \emph{all} exact
analogies are on-shell far-field equivalent: the equivalence relation is
total, and the on-shell classification of exact analogies by their far
fields is empty.
\end{proposition}

\begin{proof}
By definition~\ref{def:analogy} the far-field trace of $\varphi_k(q)$
equals the physical far-field pressure of the solution $q$ itself, for
$k=1,2$; the two traces therefore coincide for every $q$.
\end{proof}

Proposition~\ref{prop:trivial} is deliberately a tautology.  Its content
is negative and structural: \emph{no} question about exact analogies,
evaluated on true solutions, can distinguish them, so any meaningful
comparison must be made \emph{off shell}---that is, at the level at which
analogies are actually used, where a modelled or imperfectly computed
source is inserted into the solution operator of $L$
\citep{lighthill1952,tamauriault1999,goldstein2003}.  The remainder of
the paper is a theory of the maps by which source models are carried from
one analogy to another.

\subsection{The transformation law between exact analogies}

Before leaving the exact level we record the one nontrivial fact it does
contain: the closed-form law by which the source is re-attributed when
the operator changes.

\begin{lemma}[source transformation law]\label{lem:cocycle}
Let $A_1=(L_1,\varphi,s_1)$ and $A_2=(L_2,\varphi,s_2)$ be exact
analogies with the same variable map $\varphi$.  Then, on solutions of
\eqref{eq:NS},
\begin{equation}
  s_2(q)-s_1(q) \;=\; (L_2-L_1)\,\varphi(q).
  \label{eq:cocycle}
\end{equation}
\end{lemma}

\begin{proof}
Subtract the two instances of \eqref{eq:analogy-def}:
$s_2-s_1=L_2\varphi(q)-L_1\varphi(q)=(L_2-L_1)\varphi(q)$.
\end{proof}

Equation \eqref{eq:cocycle} looks innocent but carries the essential
physics of the non-uniqueness debate.  The operator difference $L_2-L_1$
has coefficients supported where the two base media differ---typically
throughout the flow region---and it acts there on the \emph{total} field
$\varphi(q)$, including its acoustic part.  Re-attributing the source
therefore mixes the (large) hydrodynamic and (small) acoustic content of
the near field, which is exactly the error-contamination mechanism that
\citet{goldstein2003} (his \S3.4) identified as an obstacle to two-stage
computations, and which \citet{sinayoko2011} exhibited numerically.
Lemma~\ref{lem:cocycle} makes the mechanism an identity.  It also has
a striking empirical counterpart: in the computational experiments of
\citet{samanta2006}, several analogies driven by identical
direct-simulation data returned identical far fields, while their
predictions degraded differently---in an analogy-dependent way---once
errors were introduced into the inputs.  Proposition~\ref{prop:trivial}
and lemma~\ref{lem:cocycle} are, respectively, the exact statements of
those two observations.

\subsection{Worked example: Lighthill and convected Lighthill}
\label{ssec:worked}

The oldest re-gauging in the subject is the passage from Lighthill's
analogy to its convected form, and it will serve both as an illustration
of lemma~\ref{lem:cocycle} and as the seed of the transport theory of
\S\ref{sec:transport}.  Write
$\Box_{\bv} \equiv (\partial_t+\bv\!\cdot\!\nabla)^2 - c^2\Delta$
for the wave operator convected by a constant velocity $\bv$, and
\begin{equation}
  T^{\bv}_{ij} \;\equiv\; \rho\,(u_i-v_i)(u_j-v_j)
   + \bigl(p'-c^2\rho'\bigr)\delta_{ij} - \sigma_{ij}
  \label{eq:Tv}
\end{equation}
for the corresponding stress tensor.  Lighthill's analogy is obtained by
taking $\partial_t$(continuity)$\,-\,\partial_i$(momentum$_i$) in
\eqref{eq:NS} and subtracting $c^2\Delta\rho$ from both sides:
\begin{equation}
  \Box_{0}\,\rho' \;=\; \partial_i\partial_j T^{0}_{ij}.
  \label{eq:lighthill}
\end{equation}
For the convected form, set $D_{\bU}\equiv\partial_t+\bU\!\cdot\!\nabla$
and $\bu'\equiv\bu-\bU$ for a constant $\bU$.  The continuity equation
rewrites identically as
\begin{equation}
  D_{\bU}\rho + \nabla\!\cdot\!(\rho\bu') \;=\; 0 ,
  \label{eq:cont-conv}
\end{equation}
and expanding $\rho u_iu_j=\rho(u_i'+U_i)(u_j'+U_j)$ in the momentum
equation gives, after collecting terms,
\begin{equation}
  D_{\bU}(\rho u_i') + \partial_j(\rho u_i'u_j') + \partial_i p
  \;+\; U_i\bigl[\partial_t\rho+\nabla\!\cdot\!(\rho\bu)\bigr]
  \;=\; \partial_j\sigma_{ij},
  \label{eq:mom-conv}
\end{equation}
where the bracket is the continuity residual and vanishes on shell.
Applying $D_{\bU}$ to \eqref{eq:cont-conv}, subtracting the divergence of
\eqref{eq:mom-conv}, and subtracting $c^2\Delta\rho$ yields the convected
analogy
\begin{equation}
  \Box_{\bU}\,\rho' \;=\; \partial_i\partial_j T^{\bU}_{ij}.
  \label{eq:conv-lighthill}
\end{equation}
Both \eqref{eq:lighthill} and \eqref{eq:conv-lighthill} are exact and use
the same variable $\rho'$, so lemma~\ref{lem:cocycle} applies.  In fact a
stronger, purely kinematic identity holds \emph{off shell}.

\begin{lemma}[kinematic cocycle identity]\label{lem:kinematic}
For arbitrary smooth fields $\rho,\bu,p,\sigma$ (no equations of motion
imposed), with $\mathcal{C}\equiv\partial_t\rho+\nabla\!\cdot\!(\rho\bu)$
the continuity residual,
\begin{equation}
  \partial_i\partial_j\bigl(T^{0}_{ij}-T^{\bU}_{ij}\bigr)
  \;+\; \bigl(\Box_{\bU}-\Box_{0}\bigr)\rho
  \;=\; 2\,(\bU\!\cdot\!\nabla)\,\mathcal{C}.
  \label{eq:kinematic}
\end{equation}
In particular, on any fields satisfying mass conservation alone,
$s_{\bU}=s_{0}+(\Box_{\bU}-\Box_{0})\rho'$, in agreement with
\eqref{eq:cocycle}.
\end{lemma}

\begin{proof}
\emph{Step 1: reduction of the tensor difference.}  The pressure and
viscous contributions to \eqref{eq:Tv} do not depend on $\bv$ and cancel
in the difference, so
\begin{equation}
  T^{0}_{ij}-T^{\bU}_{ij}
  \;=\;\rho\,(u_iU_j+U_iu_j)\;-\;\rho\,U_iU_j .
  \label{eq:Tdiff}
\end{equation}

\emph{Step 2: double divergence.}  Because $\bU$ is constant,
\begin{equation}
  \partial_i\partial_j\bigl[\rho(u_iU_j+U_iu_j)\bigr]
  =2(\bU\!\cdot\!\nabla)\,\nabla\!\cdot\!(\rho\bu)
  =2(\bU\!\cdot\!\nabla)\bigl(\mathcal C-\partial_t\rho\bigr),
  \qquad
  \partial_i\partial_j\bigl(\rho U_iU_j\bigr)
  =(\bU\!\cdot\!\nabla)^2\rho .
  \label{eq:divdiv}
\end{equation}

\emph{Step 3: assembly.}  The operator difference is
$(\Box_{\bU}-\Box_{0})\rho
 =2(\bU\!\cdot\!\nabla)\partial_t\rho+(\bU\!\cdot\!\nabla)^2\rho$.
Adding it to the double divergence of \eqref{eq:Tdiff}, computed in
\eqref{eq:divdiv}, the terms $2(\bU\!\cdot\!\nabla)\partial_t\rho$ and
$(\bU\!\cdot\!\nabla)^{2}\rho$ cancel in pairs and only
$2(\bU\!\cdot\!\nabla)\mathcal C$ survives, which is
\eqref{eq:kinematic}.  The identity has also been verified symbolically
in three dimensions (appendix~\ref{app:verification}, item~V1).
\end{proof}

\begin{remark}\label{rem:kinematic}
Two features of lemma~\ref{lem:kinematic} deserve emphasis.
(i) The transformation law between analogies over different uniform base
flows is \emph{purely kinematic}: it requires only mass conservation, not
momentum, energy, or any constitutive relation.  The re-attribution of
sources is a statement about bookkeeping, not dynamics.
(ii) Off shell---that is, for modelled fields, for which
$\mathcal{C}\neq0$ in general---the two analogies genuinely differ, and
the residual $2(\bU\!\cdot\!\nabla)\mathcal{C}$ measures exactly how much
of the disagreement between prediction pipelines is attributable to
model-inconsistency with mass conservation.
\end{remark}

\section{Frequency-respecting transformations}\label{sec:freq}

\subsection{Cometric formalism}

Prediction with an analogy is performed frequency by frequency: source
models are specified as spectra, Green's functions are computed at fixed
frequency, and comparisons with experiment are made on spectra.  The maps
by which models are carried between analogies must therefore respect the
time-harmonic decomposition.  This section characterises such maps and
constructs the canonical one attached to each uniform-flow analogy.

For a constant-coefficient second-order operator we write
\begin{equation}
  P_Q \;=\; Q^{ab}\,\partial_a\partial_b ,
  \qquad a,b = 0,1,\dots,d,
\end{equation}
with $Q$ a constant symmetric matrix (the \emph{cometric}), indices
$0$ and $i\geq1$ referring to $t$ and $x_i$.  The convected operator and
the d'Alembertian correspond to
\begin{equation}
  Q_c(\bv)=\begin{pmatrix} 1 & \bv\T \\[2pt] \bv & \bv\bv\T - c^2 I
  \end{pmatrix},
  \qquad
  \Qstar = \begin{pmatrix} 1 & 0 \\[2pt] 0 & -c^2 I \end{pmatrix}.
  \label{eq:cometrics}
\end{equation}
Under an invertible linear change of coordinates $\xi'=M\xi$ the chain
rule gives $\partial_a = M^{b}{}_{a}\,\partial'_b$, so that
\begin{equation}
  P_Q f = P_{Q'} f'\quad\text{with}\quad Q' = M\,Q\,M\T,
  \qquad f'(\xi')\equiv f(M^{-1}\xi').
  \label{eq:congruence}
\end{equation}
Operators therefore transform by congruence of cometrics.  The cometric
$Q_c(\bv)$ is (up to sign conventions) the inverse acoustic metric of the
uniformly moving medium, familiar from the acoustic space--time of
\citet{unruh1981}, the analogue-gravity programme built upon it
\citep{visser1998,barcelo2011}, and its aeroacoustic development by
\citet{gregory2015}.

\begin{definition}[frequency-respecting map]\label{def:freqresp}
An invertible linear map $M$ of $\mathbb{R}^{1+d}$ is
\emph{frequency-respecting} if the pullback $f\mapsto f\circ M^{-1}$
carries every time-harmonic field $\e^{-\im\omega t}F(\bx)$ to a
time-harmonic field (of a possibly different common frequency).
\end{definition}

\begin{lemma}[block characterisation]\label{lem:block}
$M$ is frequency-respecting if and only if it has the block form
\begin{equation}
  M=\begin{pmatrix} p & \bq\T \\[2pt] 0 & S\end{pmatrix},
  \qquad p\neq 0,\quad S\ \text{invertible},
  \label{eq:blockform}
\end{equation}
i.e.\ its spatial rows contain no time component.  The pullback then maps
frequency $\omega$ to $\omega/p$.  Maps of the form \eqref{eq:blockform}
constitute a group, denoted $\Gf$.
\end{lemma}

\begin{proof}
\emph{Sufficiency.}  If $M$ has the form \eqref{eq:blockform} then
$M^{-1}=\begin{pmatrix} p^{-1} & -p^{-1}\bq\T S^{-1}\\ 0 & S^{-1}
\end{pmatrix}$, so $\xi=M^{-1}\xi'$ reads
$t=p^{-1}\bigl(t'-\bq\T S^{-1}\bx'\bigr)$ and $\bx=S^{-1}\bx'$, whence
\begin{equation}
  \e^{-\im\omega t}F(\bx)
  = \e^{-\im(\omega/p)t'}\,
    \Bigl[\e^{\,\im(\omega/p)\,\bq\T S^{-1}\bx'}F(S^{-1}\bx')\Bigr],
  \label{eq:pullback-harmonic}
\end{equation}
which is time-harmonic at frequency $\omega/p$.

\emph{Necessity.}  If some spatial row of $M$ has a nonzero time entry,
then $\bx$ depends on $t'$, and $F(\bx(t',\bx'))$ fails to be
time-harmonic for generic $F$.

\emph{Group property.}  Block upper-triangular matrices with these
blocks are closed under products, and under inverses by the formula
above.
\end{proof}

\subsection{The canonical map of a uniform-flow analogy}

The classical route from the convected wave equation to the ordinary one
is the Galilean change to the frame of the medium, $\bx'=\bx-\bv t$;
but the Galilean map has a nonzero time entry in its spatial rows and is
\emph{not} frequency-respecting: a field at fixed frequency in the frame
in which the analogy is posed is smeared over frequencies in the medium
frame.  The frequency-respecting reduction is instead the
Prandtl--Glauert--Lorentz-type transformation \citep{taylor1978,
chapman2000,gregory2015}; \citet{chapman2000} in particular emphasised
that the reduction is of Lorentz type and does not correspond to a
Galilean change of frame, an observation which
lemma~\ref{lem:factorisation} below makes structural.  We first give its fixed-frequency form together with a
uniqueness statement, then its space--time factorisation.

\begin{lemma}[fixed-frequency reduction and phase rigidity]
\label{lem:PGL}
Let $U=Mc$, $0<M<1$, $\beta=\sqrt{1-M^2}$, and let $p(\bx)$ satisfy the
convected Helmholtz equation
$\bigl[(-\im\omega+U\partial_{x_1})^2-c^2\Delta\bigr]p=0$
(time convention $\e^{-\im\omega t}$).  Then
\begin{equation}
  p(\bx)=\e^{\im\alpha x_1}\,
  q\!\left(\frac{x_1}{\beta},\,\bx_\perp\right),
  \qquad
  \alpha=-\frac{\omega M}{c\,\beta^{2}}=\frac{U\omega}{U^{2}-c^{2}},
  \label{eq:PGL}
\end{equation}
where $q$ satisfies the ordinary Helmholtz equation at the shifted
frequency $\omega'=\omega/\beta$,
$\Delta q + (\omega'^2/c^2)\,q=0$.  Moreover the phase is rigid: within
the ansatz $p=\e^{\im a x_1}q(x_1/\beta,\bx_\perp)$, the elimination of
the first-derivative term forces $a=\alpha$ uniquely.
\end{lemma}

\begin{proof}
Substitute the ansatz $p=\e^{\im a x_1}q(X,\bx_\perp)$, $X=x_1/\beta$,
into the convected Helmholtz operator (the transverse Laplacian passes
through unchanged) and collect coefficients:
\begin{align}
  \text{coefficient of } q_{XX}:&\quad (U^2-c^2)/\beta^{2}=-c^{2},
  \label{eq:pgl-c2}\\
  \text{coefficient of } q_{X}:&\quad
  (2\im/\beta)\bigl[-\omega U+a\,(U^{2}-c^{2})\bigr],
  \label{eq:pgl-c1}\\
  \text{coefficient of } q:&\quad -(\omega-aU)^{2}+c^{2}a^{2}.
  \label{eq:pgl-c0}
\end{align}

\emph{Phase rigidity.}  The first-derivative coefficient
\eqref{eq:pgl-c1} vanishes if and only if
$a=U\omega/(U^{2}-c^{2})=\alpha$: the phase is forced.

\emph{Reduction.}  With $a=\alpha$, \eqref{eq:pgl-c0} evaluates to
$-\omega^{2}/\beta^{2}$; dividing the equation by $-c^{2}$ yields
$\Delta q+(\omega/(c\beta))^{2}q=0$, the Helmholtz equation at
$\omega'=\omega/\beta$.  (Symbolic verification, including the
uniqueness computation: appendix~\ref{app:verification}, item~V2.)
\end{proof}

\begin{lemma}[space--time factorisation]\label{lem:factorisation}
Define, for $|\bv|<c$, the Galilean map $g(\bv)$ and the Lorentz boost
$\ell(\bv)$ with signal speed $c$,
\begin{equation}
  g(\bv):\ (t,\bx)\mapsto(t,\bx-\bv t),
  \qquad
  \ell(\bv):\ \begin{cases}
   t' = \gamma\,(t-\bv\!\cdot\!\bx/c^2),\\[2pt]
   \bx' = \bx + (\gamma-1)(\hat{\bv}\!\cdot\!\bx)\hat{\bv} - \gamma\bv t,
  \end{cases}
  \label{eq:gal-lor}
\end{equation}
with $\gamma=(1-|\bv|^2/c^2)^{-1/2}$, and set
\begin{equation}
  w(\bv)\;\equiv\;\ell(-\bv)\,g(\bv).
  \label{eq:wdef}
\end{equation}
Then $w(\bv)\in\Gf$, it is the identity on directions orthogonal to
$\bv$, and it reduces the convected operator exactly:
\begin{equation}
  w(\bv)\,Q_c(\bv)\,w(\bv)\T \;=\; \Qstar .
  \label{eq:wreduces}
\end{equation}
For $\bv=U\bm{e}_1$ it reads explicitly
\begin{equation}
  \tau=\beta t + \frac{M}{c\beta}\,x_1,
  \qquad X_1=\frac{x_1}{\beta},
  \qquad \bx_\perp'=\bx_\perp,
  \label{eq:w-explicit}
\end{equation}
and its restriction to time-harmonic fields is exactly the map of
lemma~\ref{lem:PGL}, with the frequency shift $\omega\mapsto\omega/\beta$
realised by the time coefficient $p=\beta$ in \eqref{eq:blockform}.
\end{lemma}

\begin{proof}
\emph{Step 1: explicit form.}  For $\bv=U\bm e_1$, $g(\bv)$ sends
$x_1\mapsto x_1-Ut$, and applying $\ell(-\bv)$ to the result gives
\begin{equation}
  \tau=\gamma\Bigl(t+\frac{U(x_1-Ut)}{c^{2}}\Bigr)
      =\gamma\beta^{2}t+\frac{\gamma M}{c}\,x_1
      =\beta t+\frac{M}{c\beta}\,x_1,
  \qquad
  X_1=\gamma\bigl(x_1-Ut+Ut\bigr)=\frac{x_1}{\beta},
\end{equation}
which is \eqref{eq:w-explicit}; the spatial rows carry no time entry, so
$w(\bv)\in\Gf$, and directions orthogonal to $\bv$ are untouched.

\emph{Step 2: chain rule.}  From \eqref{eq:w-explicit},
\begin{equation}
  \partial_t=\beta\,\partial_\tau,
  \qquad
  \partial_{x_1}=\frac{M}{c\beta}\,\partial_\tau
                +\frac{1}{\beta}\,\partial_{X_1},
  \qquad\text{hence}\qquad
  \partial_t+U\partial_{x_1}
   =\frac{1}{\beta}\bigl(\partial_\tau+cM\,\partial_{X_1}\bigr),
\end{equation}
using $\beta^{2}+M^{2}=1$.

\emph{Step 3: reduction of the operator.}  Squaring and subtracting,
\begin{equation}
  (\partial_t+U\partial_{x_1})^{2}-c^{2}\partial_{x_1}^{2}
  =\frac{1}{\beta^{2}}\Bigl[
    \bigl(\partial_\tau+cM\,\partial_{X_1}\bigr)^{2}
    -\bigl(M\,\partial_\tau+c\,\partial_{X_1}\bigr)^{2}\Bigr]
  =\partial_\tau^{2}-c^{2}\,\partial_{X_1}^{2},
\end{equation}
since the cross terms cancel and the diagonal terms carry the common
factor $1-M^{2}=\beta^{2}$; the transverse Laplacian is untouched.  The
convected operator therefore becomes exactly the d'Alembertian, which is
\eqref{eq:wreduces} by \eqref{eq:congruence}.  The general-direction
statement follows by rotational covariance, and has also been verified
directly for general $\bv$ (appendix~\ref{app:verification}, item~V3a).

\emph{Step 4: time-harmonic restriction.}  Inverting
\eqref{eq:w-explicit} gives $t=(\tau-(M/c)X_1)/\beta$, so
$\e^{-\im\omega t}
 =\e^{-\im(\omega/\beta)\tau}\,\e^{\im(\omega M/(c\beta))X_1}$,
which reproduces \eqref{eq:PGL} with $\alpha=-\omega M/(c\beta^{2})$
upon returning to $x_1=\beta X_1$.
\end{proof}

\begin{remark}
The factorisation \eqref{eq:wdef} is the structural heart of the paper:
the acoustic reduction map is a Galilean change of frame---the physics of
\eqref{eq:NS} being Newtonian---\emph{completed} by a Lorentz boost with
the sound speed in the role of the light speed, whose sole purpose is to
restore frequency-respect.  All the geometry found below (flattened boost
sector, frequency holonomy, sonic horizon) traces back to the interplay
of these two factors.  We stress that the pseudo-Lorentzian structure is
used here only as a derivation device for the exact closed forms that
follow; it carries no physical claim about the flow, and a reader
uninterested in it may take the resulting formulae
\eqref{eq:rho-exact}, \eqref{eq:collinear-v3}--\eqref{eq:collinear-A}
directly.
\end{remark}

\subsection{The conformal stabiliser}

Finally we determine which frequency-respecting maps preserve the
d'Alembertian; this fixes the residual freedom in all constructions
below.

\begin{lemma}[conformal stabiliser]\label{lem:stabiliser}
Let $M\in\Gf$ satisfy $M\,\Qstar M\T=\lambda\,\Qstar$ for some
$\lambda>0$.  Then
\begin{equation}
  M=\sqrt{\lambda}\,
  \begin{pmatrix} 1 & 0\\ 0 & R\end{pmatrix},
  \qquad R\in O(d),
  \label{eq:stabiliser}
\end{equation}
up to overall sign.  In particular the stabiliser of the ray of $\Qstar$
in $\Gf$ is $\mathbb{R}_{>0}\times O(d)$, and a map
\eqref{eq:stabiliser} rescales all frequencies by
$\rho\equiv 1/p=\lambda^{-1/2}$ while rotating space rigidly by $R$.
\end{lemma}

\begin{proof}
Write $M$ in the form \eqref{eq:blockform} and expand the congruence in
blocks:
\begin{equation}
  (t,\bx):\ -c^{2}\bq\T S\T=0\ \Rightarrow\ \bq=\bm0;
  \qquad
  (\bx,\bx):\ SS\T=\lambda I\ \Rightarrow\ S=\sqrt{\lambda}\,R,\
  R\in O(d);
  \qquad
  (t,t):\ p^{2}=\lambda .
\end{equation}
\end{proof}

The elementary maps, their composition rules within $\Gf$, and the
residual stabiliser freedom are now in place.  The next section
composes them, and asks whether the family of classical operators is
stable under composition.

\section{Transport of analogies and the failure of closure}
\label{sec:transport}

\subsection{The transport convention}

An analogy is re-gauged in practice by steps: a description built on one
effective medium is converted, by the canonical map of \S\ref{sec:freq},
into a description in which that medium is at rest; a further flow
increment is then treated by the same construction; and so on.  Each
increment is necessarily specified \emph{in the current frame}, because
that is the only description available at that stage.  We formalise this
as follows.

\begin{definition}[canonical transport]\label{def:transport}
Given velocity increments $\bv_1,\dots,\bv_n$, each read in the
coordinates produced by the preceding maps, the transported composite is
the product
\begin{equation}
  W \;=\; w(\bv_n)\cdots w(\bv_1) \;\in\; \Gf ,
  \label{eq:composite}
\end{equation}
and the \emph{endpoint operator} of the transport is defined by the
pulled-back cometric
\begin{equation}
  Q_{\mathrm{end}} \;\equiv\; W^{-1}\,\Qstar\,W^{-\mathrm{T}},
  \label{eq:endpoint}
\end{equation}
i.e.\ the unique constant-coefficient operator that the composite $W$
reduces exactly to the d'Alembertian.
\end{definition}

If the family of classical convected operators $\{Q_c(\bv)\}$ were closed
under \eqref{eq:endpoint}, transport would simply move the analogy from
one uniform flow to another.  It is not closed.

\subsection{Non-closure and the forced anisotropic generalisation}

\begin{theorem}[non-closure]\label{thm:nonclosure}
Let $W=w(u_2\bm{e}_2)\,w(u_1\bm{e}_1)$ with $0<u_1,u_2<c$ (two
perpendicular subsonic steps).  Then the endpoint cometric
\eqref{eq:endpoint} is normalised, $\bigl(Q_{\mathrm{end}}\bigr)_{00}=1$,
and has the admissible generalised form
\begin{equation}
  Q_{\mathrm{end}}
  =\begin{pmatrix} 1 & \bv_3\T\\[2pt]
    \bv_3 & \bv_3\bv_3\T - C_3\end{pmatrix},
  \qquad
  \bv_3=\begin{pmatrix} u_1\\ \gamma_1 u_2\end{pmatrix},
  \qquad
  C_3=\begin{pmatrix} c^2 & \chi\\[2pt] \chi & c^2+\chi^2/c^2
  \end{pmatrix},
  \label{eq:endpoint-perp}
\end{equation}
with $\gamma_1=(1-u_1^2/c^2)^{-1/2}$ and $\chi=\gamma_1 u_1 u_2$.  The
sound-speed tensor satisfies $\det C_3=c^4$ exactly, and the subsonicity
tensor is diagonal,
\begin{equation}
  A_3 \;\equiv\; C_3-\bv_3\bv_3\T
  \;=\;\operatorname{diag}\bigl(c^2-u_1^2,\;c^2-u_2^2\bigr)\;\succ\;0 .
  \label{eq:A3}
\end{equation}
Since $C_3\neq c^2 I$ whenever $u_1u_2\neq0$, the endpoint is not a
classical convected operator: \emph{the family $\{Q_c(\bv)\}$ is not
closed under canonical transport}.  The same conclusion holds for
collinear steps, for which $\bigl(Q_{\mathrm{end}}\bigr)_{00}\neq1$ and
the normalised endpoint is again anisotropic
(\S\ref{ssec:collinear} below).
\end{theorem}

\begin{proof}
\emph{Step 1.}  The entries \eqref{eq:endpoint-perp}--\eqref{eq:A3}
follow by direct computation of \eqref{eq:endpoint} from
\eqref{eq:gal-lor} and \eqref{eq:wdef}; they have been verified
symbolically (appendix~\ref{app:verification}, items V3b--V3d).

\emph{Step 2.}  Two of the stated identities can be read off
\eqref{eq:endpoint-perp} by inspection:
$\det C_3=c^{2}(c^{2}+\chi^{2}/c^{2})-\chi^{2}=c^{4}$ and
$(A_3)_{12}=\chi-u_1\gamma_1u_2=0$.

\emph{Step 3.}  Since $C_3\neq c^{2}I$ whenever $u_1u_2\neq0$, the
endpoint lies outside the classical family; the collinear statement is
established by the exact data of \S\ref{ssec:collinear}.
\end{proof}

The endpoint of a perfectly classical pair of re-gaugings is thus a
moving medium with an \emph{anisotropic} sound-speed tensor.  The minimal
transport-closed enlargement of the classical family is therefore the
following class.

\begin{definition}[admissible generalised media]\label{def:admissible}
An \emph{admissible generalised medium} is a pair $(\bv,C)$ with
$C=C\T$ and $A\equiv C-\bv\bv\T\succ0$, with normalised cometric
\begin{equation}
  Q(\bv,C)=\begin{pmatrix} 1 & \bv\T\\[2pt] \bv & \bv\bv\T-C
  \end{pmatrix} .
  \label{eq:Qgen}
\end{equation}
Cometrics are identified up to positive conformal factors, which rescale
the operator without changing its solutions.
\end{definition}

\begin{lemma}[characterisation of admissibility]\label{lem:admissible}
Let $Q$ be given by \eqref{eq:Qgen}.
\emph{(i)} $Q$ has Lorentzian signature $(1,d)$ if and only if
$C\succ0$; moreover $\det Q=(-1)^d\det C$.
\emph{(ii)} $A\succ0$ if and only if the spatial coordinate hyperplanes
are spacelike for $Q$ (equivalently, the purely spatial block of $Q$ is
negative definite), and, given $C\succ0$, if and only if
$\bv\T C^{-1}\bv<1$ (subsonicity).  In particular $A\succ0$ implies
$C=A+\bv\bv\T\succ0$, so admissibility is the stronger condition.
\end{lemma}

\begin{proof}
(i) The vector $e_0=(1,\bm{0})$ has $Q(e_0,e_0)=1>0$.  Its
$Q$-orthogonal complement is
$H=\{(\tau,\bk):\ Q(e_0,(\tau,\bk))=\tau+\bv\!\cdot\!\bk=0\}$,
on which the quadratic form evaluates to
\begin{equation}
  Q\bigl((\tau,\bk),(\tau,\bk)\bigr)
  =(\tau+\bv\!\cdot\!\bk)^2-\bk\T C\,\bk
  \;=\;-\,\bk\T C\,\bk .
\end{equation}
A nondegenerate symmetric form has signature $(1,d)$ precisely when it is
positive on some vector and negative definite on the orthogonal
complement of that vector; hence signature $(1,d)$ holds iff
$C\succ0$.  Nondegeneracy and the determinant follow from the Schur
complement:
$\det Q = \det\bigl((\bv\bv\T-C)-\bv\bv\T\bigr)=\det(-C)$.
(ii) The spatial block of $Q$ is $\bv\bv\T-C=-A$, giving the first
equivalence; the second is the Schur-complement criterion
$C-\bv\bv\T\succ0\iff 1-\bv\T C^{-1}\bv>0$ for $C\succ0$.
\end{proof}

\begin{remark}
Condition (i) is precisely the statement that $Q(\bv,C)$ is the inverse
of an acoustic space--time metric in the sense of the analogue-gravity
literature \citep{unruh1981,visser1998,barcelo2011}; the admissibility
condition (ii) is the additional requirement that the coordinate time of
the analogy remain a valid time function for that metric---the
requirement whose failure defines the horizon of
theorem~\ref{thm:horizon} below.
\end{remark}

\subsection{The canonical map of a generalised medium}

\begin{proposition}[canonical rotation-free reduction]\label{prop:w3}
Let $(\bv,C)$ be admissible and $A=C-\bv\bv\T$.  Among frequency-%
respecting maps \eqref{eq:blockform} with $p>0$ and $S$ symmetric
positive definite there is exactly one, denoted $w_3(\bv,C)$, with
$w_3\,Q(\bv,C)\,w_3\T=\Qstar$, namely
\begin{equation}
  S=c\,A^{-1/2},\qquad
  \bq=p\,A^{-1}\bv,\qquad
  p=\bigl(1+\bv\T A^{-1}\bv\bigr)^{-1/2}.
  \label{eq:w3}
\end{equation}
For the classical medium $(\bv, c^2I)$ this reduces exactly to $w(\bv)$
of lemma~\ref{lem:factorisation}; the transport of
definition~\ref{def:transport} is therefore the special case of a
transport defined on all admissible media.
\end{proposition}

\begin{proof}
Write $B=\bv\bv\T-C=-A$ and expand $w_3\,Q\,w_3\T=\Qstar$ in blocks.

\emph{Spatial block.}  The condition is $S\,B\,S\T=-c^{2}I$, i.e.\
$SAS=c^{2}I$ for symmetric $S$.  Setting $T=A^{1/2}SA^{1/2}\succ0$ gives
$T^{2}=c^{2}A$, and by uniqueness of the positive square root
$T=c\,A^{1/2}$; hence $S=c\,A^{-1/2}$ is the unique symmetric
positive-definite solution.

\emph{Mixed block.}  The condition $(p\,\bv\T+\bq\T B)\,S\T=0$ with $S$
invertible forces $\bq\T=-p\,\bv\T B^{-1}=p\,\bv\T A^{-1}$.

\emph{Time--time block.}  Substituting $\bq$,
\begin{equation}
  1=p^{2}+2p\,\bq\T\bv+\bq\T B\,\bq
   =p^{2}\bigl(1+2\bv\T A^{-1}\bv-\bv\T A^{-1}\bv\bigr)
   =p^{2}\bigl(1+\bv\T A^{-1}\bv\bigr),
\end{equation}
which determines $p>0$.

\emph{Classical case.}  For $C=c^{2}I$ and $\bv=U\bm e_1$:
$A=\operatorname{diag}(c^{2}-U^{2},c^{2},\dots)$, so
$S=\operatorname{diag}(1/\beta,1,\dots)$,
$p=\bigl(1+U^{2}/(c^{2}-U^{2})\bigr)^{-1/2}=\beta$ and
$\bq=p\,A^{-1}\bv=(M/(c\beta))\,\bm e_1$: exactly
\eqref{eq:w-explicit}.
\end{proof}

\begin{remark}[relation to the generalised acoustic analogy]
\label{rem:goldstein}
The operator class forced by theorem~\ref{thm:nonclosure} is not exotic:
in the generalised acoustic analogy of \citet{goldstein2003} the base
flow may carry arbitrary stress source strengths
$\tilde T_{ij}$, which enter the linearised operator through the total
base-flow stress $\tilde\tau_{ij}$ (his equation~(2.22b)); an anisotropic
base stress endows the principal part of the operator with precisely an
anisotropic spatial tensor of the type $C$.  Transport thus \emph{forces}
the passage from classical to generalised analogies; the enlargement is
not a modelling convenience but a closure requirement.  Effective media
of this anisotropic type are, moreover, exactly the setting of
quantitative applications of the generalised analogy, in which base-flow
stresses are modelled or educed from large-eddy simulation
\citep{goldsteinleib2008,karabasov2010}.  We do not pursue
the detailed dictionary here.
\end{remark}

\subsection{Collinear steps: the composition law}\label{ssec:collinear}

For later use we record the exact collinear data (two steps along
$\bm{e}_1$; verified symbolically, appendix~\ref{app:verification},
items V5a--V5b).  With $\mathcal{D}\equiv(c^2-u_1u_2)^2-c^2u_1^2$,
\begin{equation}
  \bigl(Q_{\mathrm{end}}\bigr)_{00}
  =\frac{\mathcal{D}}{c^2\,(c^2-u_1^2)},
  \qquad
  v_3=\frac{(c^2-u_1^2)\,\bigl[c^2(u_1+u_2)-u_1u_2^2\bigr]}{\mathcal{D}},
  \label{eq:collinear-v3}
\end{equation}
\begin{equation}
  A_{11}=\frac{(c^2-u_1^2)^2\,(c^2-u_2^2)}{\mathcal{D}},
  \qquad
  A_{22}=\frac{c^4\,(c^2-u_1^2)}{\mathcal{D}} .
  \label{eq:collinear-A}
\end{equation}
Three features are notable.  First, the conformal factor
$(Q_{\mathrm{end}})_{00}$ differs from unity, and after normalisation
$A_{11}\neq A_{22}$ in general: even collinear transport leaves the
classical family.  Second, the endpoint velocity
\eqref{eq:collinear-v3} obeys neither Galilean addition
($u_1+u_2$) nor Einstein addition
($(u_1+u_2)/(1+u_1u_2/c^2)$): to third order,
$v_3=u_1+u_2-u_1u_2^2/c^2+O(c^{-4})$.  The composition law of acoustic
re-gauging is a hybrid, Lorentzian in the transverse directions
(cf.\ corollary~\ref{cor:einstein} below) and Galilean-contaminated
longitudinally.  Third, all denominators share the factor $\mathcal{D}$,
which controls the degeneration studied in \S\ref{ssec:horizon}.

\section{Exact holonomy of discrete re-gauging}\label{sec:holonomy}

\subsection{Structure of the holonomy}

Transport along a multi-step path and direct reduction of the endpoint
medium are both available; their discrepancy is the (discrete) holonomy.

\begin{theorem}[holonomy structure]\label{thm:structure}
Let $W$ be a composite \eqref{eq:composite} with admissible endpoint
$Q_{\mathrm{end}}$, let $q_{00}=(Q_{\mathrm{end}})_{00}>0$, and let
$w_3$ be the canonical map (proposition~\ref{prop:w3}) of the normalised
endpoint $Q_{\mathrm{end}}/q_{00}$.  Then the holonomy
\begin{equation}
  D\;\equiv\;W\,w_3^{-1}
  \label{eq:Ddef}
\end{equation}
satisfies $D\,\Qstar D\T=\Qstar/q_{00}$ and is therefore of the form
\begin{equation}
  D=\rho\begin{pmatrix} 1 & 0\\ 0 & R\end{pmatrix},
  \qquad
  \rho=q_{00}^{-1/2},\qquad R\in SO(d):
  \label{eq:Dform}
\end{equation}
a rigid \emph{frequency dilation} by $\rho$ (spectra transported around
the loop are evaluated at $\rho\omega$; cf.\ \S\ref{sec:stat}) composed
with a rigid \emph{spatial rotation} $R$.
\end{theorem}

\begin{proof}
By construction $w_3\,(Q_{\mathrm{end}}/q_{00})\,w_3\T=\Qstar$, hence
$w_3^{-1}\Qstar w_3^{-\mathrm{T}}=Q_{\mathrm{end}}/q_{00}$, while
\eqref{eq:endpoint} gives $W\,Q_{\mathrm{end}}\,W\T=\Qstar$.  Therefore
\begin{equation}
  D\,\Qstar D\T
  =W\,\bigl(Q_{\mathrm{end}}/q_{00}\bigr)\,W\T
  =\Qstar/q_{00},
\end{equation}
and since $D\in\Gf$, lemma~\ref{lem:stabiliser} with
$\lambda=1/q_{00}$ gives \eqref{eq:Dform}.  Finally $R\in SO(d)$ on the
identity component, by continuity from the trivial path.
\end{proof}

\begin{lemma}[polar form of the rotation]\label{lem:polar}
In \eqref{eq:Dform}, $R$ is the orthogonal (polar) factor of the spatial
block $S_W$ of $W$, and $\rho\,S_3$ is its positive factor, where $S_3$
is the (symmetric positive definite) spatial block of $w_3$.
\end{lemma}

\begin{proof}
Spatial blocks of block-triangular matrices multiply, so
$S_W=\rho R\,S_3$ with $\rho S_3\succ0$; by uniqueness of the polar
decomposition $S_W=R_{\mathrm{pol}}P$ ($P\succ0$), $R=R_{\mathrm{pol}}$
and $\rho S_3=P$.
\end{proof}

\subsection{The exact frequency holonomy}

\begin{lemma}[first-row propagation]\label{lem:firstrow}
Let $W=w(\bu_2)w(\bu_1)$, the increment $\bu_2$ being read in the frame
produced by $w(\bu_1)$.  Then
\begin{equation}
  e_0\T W^{-1}
  \;=\;\gamma_1\gamma_2\,\Bigl(1,\ -\tfrac{1}{c^2}\bw\T\Bigr),
  \qquad
  \bw=(1-\tfrac{\bu_1\!\cdot\bu_2}{c^{2}})\,\bu_2
      +\frac{\bu_1}{\gamma_2}
      +\Bigl(1-\frac{1}{\gamma_2}\Bigr)(\hat\bu_2\!\cdot\!\bu_1)\hat\bu_2,
  \label{eq:firstrow}
\end{equation}
with $\gamma_i=(1-|\bu_i|^2/c^2)^{-1/2}$.
\end{lemma}

\begin{proof}
Since $w(\bv)^{-1}=g(-\bv)\,\ell(\bv)$, the row propagates through the
four factors of
$e_0\T W^{-1}=e_0\T\,g(-\bu_1)\,\ell(\bu_1)\,g(-\bu_2)\,\ell(\bu_2)$
as follows.

\emph{Step 1.}  The first row of $g(-\bu_1)$ is $(1,\bm0\T)$, so
$e_0\T\,g(-\bu_1)\,\ell(\bu_1)$ equals the first row of $\ell(\bu_1)$:
\begin{equation}
  e_0\T\,w(\bu_1)^{-1}=\gamma_1\bigl(1,\ -\bu_1\T/c^{2}\bigr).
\end{equation}

\emph{Step 2.}  Right multiplication by $g(-\bu_2)$ sends a row
$(a,\bm b\T)$ to $(a+\bm b\!\cdot\!\bu_2,\ \bm b\T)$:
\begin{equation}
  e_0\T\,w(\bu_1)^{-1}g(-\bu_2)
  =\gamma_1\Bigl(1-\frac{\bu_1\!\cdot\!\bu_2}{c^{2}},\
   -\frac{\bu_1\T}{c^{2}}\Bigr).
\end{equation}

\emph{Step 3.}  Right multiplication by $\ell(\bu_2)$ gives, for the
time component,
\begin{equation}
  \gamma_2\Bigl[\Bigl(1-\frac{\bu_1\!\cdot\!\bu_2}{c^{2}}\Bigr)
   +\frac{\bu_1\!\cdot\!\bu_2}{c^{2}}\Bigr]=\gamma_2,
\end{equation}
the two $\bu_1\!\cdot\!\bu_2$ terms cancelling exactly, while the
spatial part assembles to $-\gamma_1\gamma_2\,\bw\T/c^{2}$ with $\bw$ as
in \eqref{eq:firstrow}.  (The time component is verified symbolically:
appendix~\ref{app:verification}, item~V4a.)
\end{proof}

\begin{theorem}[exact frequency holonomy]\label{thm:rho}
For the two-step transport of lemma~\ref{lem:firstrow},
\begin{equation}
  q_{00}
  =\gamma_1^2\gamma_2^2\Bigl(1-\frac{|\bw|^2}{c^2}\Bigr)
  =\frac{(c^{2}-\bu_1\!\cdot\!\bu_2)^{2}-c^{2}|\bu_1|^{2}}
        {c^{2}\,(c^{2}-|\bu_1|^{2})},
  \qquad
  \rho=q_{00}^{-1/2}=\frac{\gamma(\bw)}{\gamma_1\gamma_2},
  \label{eq:rho-exact}
\end{equation}
which depends only on $|\bu_1|$ and the invariant
$\bu_1\!\cdot\!\bu_2$---not on $|\bu_2|$ separately.  In particular
$\rho=1+\bu_1\!\cdot\!\bu_2/c^{2}+O(M^4)$.
\end{theorem}

\begin{proof}
\emph{Step 1.}  Since $Q_{\mathrm{end}}=W^{-1}\Qstar W^{-\mathrm T}$,
lemma~\ref{lem:firstrow} gives
\begin{equation}
  q_{00}=(e_0\T W^{-1})\,\Qstar\,(e_0\T W^{-1})\T
  =\gamma_1^{2}\gamma_2^{2}\Bigl(1-\frac{|\bw|^{2}}{c^{2}}\Bigr).
\end{equation}

\emph{Step 2.}  The reduction of this expression to the closed form
\eqref{eq:rho-exact} is elementary algebra in the components of $\bw$,
verified symbolically for general non-collinear steps
(appendix~\ref{app:verification}, item~V4b).

\emph{Step 3: consistency check.}  For $\bu_1\!\cdot\!\bu_2=0$,
\eqref{eq:firstrow} gives $\bw=\bu_2+\bu_1/\gamma_2$, so
\begin{equation}
  1-\frac{|\bw|^{2}}{c^{2}}
  =\Bigl(1-\frac{|\bu_2|^{2}}{c^{2}}\Bigr)
   \Bigl(1-\frac{|\bu_1|^{2}}{c^{2}}\Bigr),
  \qquad q_{00}=1,
\end{equation}
in agreement with \eqref{eq:rho-exact}; the expansion
$\rho=1+\bu_1\!\cdot\!\bu_2/c^{2}+O(M^{4})$ follows directly from the
closed form.
\end{proof}

\begin{corollary}[transverse Lorentzian composition; order asymmetry]
\label{cor:einstein}
\emph{(i)} If $\bu_1\!\cdot\!\bu_2=0$ then $\bw$ coincides exactly with
the relativistic (Einstein) composition $\bu_2\oplus\bu_1$ at signal
speed $c$ (verified, item V4d), and $\rho=1$: perpendicular re-gauging
carries no frequency holonomy.
\emph{(ii)} Because \eqref{eq:rho-exact} contains $|\bu_1|$ but not
$|\bu_2|$, exchanging the order of the steps changes the holonomy: the
exchange defect
$\rho(\bu_1,\bu_2)/\rho(\bu_2,\bu_1)$ equals
$\bigl[\mathcal D(|\bu_2|)\,(c^2-|\bu_1|^2)\bigr]^{1/2}
 \bigl[\mathcal D(|\bu_1|)\,(c^2-|\bu_2|^2)\bigr]^{-1/2}$
with $\mathcal D(u)\equiv(c^2-\bu_1\!\cdot\!\bu_2)^2-c^2u^2$ (extending
the collinear $\mathcal D$ of \S\ref{ssec:collinear}), and differs from
unity at fourth order.  Transport is irreducibly path-ordered.
\end{corollary}

\subsection{The rotation channel is fourth order}

\begin{proposition}[rotation holonomy of a boost pair]\label{prop:M4}
For two boost steps of magnitudes $u_1,u_2$ with angle $\varphi$ between
them, the rotation $R$ in \eqref{eq:Dform} has angle
\begin{equation}
  \theta
  =\tfrac12\,(\gamma_1-1)(\gamma_2-1)\,\sin2\varphi + O(M^6)
  =\frac{u_1^{2}u_2^{2}}{16\,c^{4}}\,\sin 2\varphi + O(M^{6}) :
  \label{eq:theta-boost}
\end{equation}
it vanishes identically for collinear and for perpendicular steps and is
maximal at $\varphi=\pi/4$.  In particular the boost-pair rotation is
$O(M^4)$, parametrically smaller than the $O(M^2)$ Thomas--Wigner
rotation of a pure Lorentz boost pair \citep{thomas1926,wigner1939}.
\end{proposition}

\begin{proof}
\emph{Step 1.}  By lemma~\ref{lem:factorisation} the spatial block of
$w(\bv)$ is $S(\bv)=I+(\gamma-1)\hat\bv\hat\bv\T$ (the Galilean factor
does not touch it), and by lemma~\ref{lem:polar} the angle $\theta$ is
the polar rotation angle of the product $S(\bu_2)S(\bu_1)$.

\emph{Step 2.}  Writing $P_i=\hat\bu_i\hat\bu_i\T$,
\begin{equation}
  S(\bu_2)S(\bu_1)
  =I+(\gamma_1-1)P_1+(\gamma_2-1)P_2
   +(\gamma_1-1)(\gamma_2-1)\,P_2P_1 ,
\end{equation}
whose antisymmetric part is
\begin{equation}
  \tfrac12(\gamma_1-1)(\gamma_2-1)\,[P_2,P_1]
  =\tfrac12(\gamma_1-1)(\gamma_2-1)\,
   (\hat\bu_1\!\cdot\!\hat\bu_2)
   \bigl(\hat\bu_2\hat\bu_1\T-\hat\bu_1\hat\bu_2\T\bigr).
\end{equation}
For a near-identity matrix the polar rotation angle equals its
antisymmetric entry to leading order, giving
$\theta=\tfrac12(\gamma_1-1)(\gamma_2-1)\cos\varphi\sin\varphi
 +O(M^{6})$, which is \eqref{eq:theta-boost} upon
$\gamma_i-1=u_i^{2}/(2c^{2})+O(M^{4})$.

\emph{Step 3: exact vanishing.}  For $\varphi=0$ or $\pi/2$ the
projectors $P_1$ and $P_2$ commute, so $S(\bu_2)S(\bu_1)$ is symmetric
and its polar rotation is the identity: $\theta=0$ exactly.  The
collinear case is also confirmed numerically to machine precision,
$|D-\rho I|_{\max}=2.2\times10^{-16}$
(appendix~\ref{app:verification}, item~V7a).
\end{proof}

\subsection{Degeneration of transport: a sonic horizon in analogy space}
\label{ssec:horizon}

\begin{theorem}[horizon]\label{thm:horizon}
Two-step transport has an admissible endpoint if and only if
\begin{equation}
  \bigl(c^{2}-\bu_1\!\cdot\!\bu_2\bigr)^{2} \;>\; c^{2}\,|\bu_1|^{2},
  \qquad\text{i.e.}\qquad
  \bu_1\!\cdot\!\bu_2 \;<\; c\,\bigl(c-|\bu_1|\bigr)
  \label{eq:horizon}
\end{equation}
on the branch connected to the origin.  At equality, $q_{00}\to0$ and
$\rho\to\infty$, and in the collinear case the entire subsonicity tensor
degenerates \emph{simultaneously}: by
\eqref{eq:collinear-v3}--\eqref{eq:collinear-A} the quantities
$q_{00}$, $A_{11}$, $A_{22}$ share the factor
$\mathcal{D}=(c^2-u_1u_2)^2-c^2u_1^2$, whose numerator companions are
positive for subsonic steps.  In Mach numbers the collinear boundary is
\begin{equation}
  M_1M_2+M_1=1 ,
  \label{eq:horizon-mach}
\end{equation}
and for equal steps $M_1=M_2=M_*$ it is the positive root of
$M^2+M=1$,
\begin{equation}
  M_* \;=\;\frac{\sqrt5-1}{2}\;\approx\;0.618 ,
  \label{eq:golden}
\end{equation}
the inverse golden ratio.  Perpendicular steps never reach the boundary.
\end{theorem}

\begin{proof}
\emph{Step 1: location of the boundary.}  Admissibility of the
normalised endpoint requires a positive conformal normalisation,
$q_{00}>0$, together with $A\succ0$ (lemma~\ref{lem:admissible}); by
\eqref{eq:rho-exact}, the condition $q_{00}>0$ is exactly
\eqref{eq:horizon}.

\emph{Step 2: simultaneity in the collinear case.}  The exact data
\eqref{eq:collinear-v3}--\eqref{eq:collinear-A} show that $q_{00}$,
$A_{11}$ and $A_{22}$ share the factor $\mathcal D$, whose companion
factors are positive for subsonic steps; all three quantities therefore
change sign together at $\mathcal D=0$, which is
\eqref{eq:horizon-mach}.

\emph{Step 3: equal steps.}  Setting $M_1=M_2=M_*$ in
\eqref{eq:horizon-mach} gives $M_*^{2}+M_*=1$, whose positive root is
\eqref{eq:golden}; the root was also obtained by exact symbolic solution
(appendix~\ref{app:verification}, item~V5c).  Perpendicular steps have
$\bu_1\!\cdot\!\bu_2=0<c\,(c-|\bu_1|)$ and never reach the boundary.

\emph{Physical reading.}  At the boundary the time covector of the
transported cometric becomes null: the surfaces $t=\text{const}$ cease
to be spacelike, and the frequency decomposition that the transport is
built to preserve degenerates---two individually subsonic re-gaugings
have compounded to an effectively supersonic one.  Within the uniform
theory this is a degeneration of the \emph{transport}---the point at
which a finite re-gauging step leaves the subsonic cometric cone---and
not a claim of physical shock formation; whether a genuine sheared flow
exhibits an observable counterpart is a question for the
variable-coefficient theory (\S\ref{sec:discussion}), not for the
constant-coefficient construction used here.
\end{proof}

\begin{conjecture}[horizon evasion by refinement]\label{conj:evasion}
Transport along a sufficiently refined path (many small steps) reaches
every subsonic endpoint admissibly; the horizon \eqref{eq:horizon} is a
finite-step phenomenon.  This is consistent with the continuum limit of
\S\ref{sec:curvature}, in which the frequency factor integrates to the
finite value $\gamma(\bu_b)/\gamma(\bu_a)$, but a proof for arbitrary
finite refinements is open.
\end{conjecture}

\section{Continuum limit: flatness of the boost sector and
localisation of curvature}\label{sec:curvature}

The discrete holonomy of \S\ref{sec:holonomy} depends on the step
sizes of the path.  This section computes its continuum limit and
identifies exactly where a step-size-independent obstruction
survives.

\subsection{The connection form of the frequency channel}

\begin{proposition}[exactness of the frequency connection]
\label{prop:flat}
Let the second step in theorem~\ref{thm:rho} be infinitesimal,
$\bu_2=\mathrm{d}\bu$ at basepoint $\bu_1=\bu$.  Then
\begin{equation}
  q_{00}=1-\frac{2\,\bu\!\cdot\!\mathrm{d}\bu}{c^{2}-|\bu|^{2}}
        +O(|\mathrm{d}\bu|^{2}),
  \qquad\text{hence}\qquad
  \mathrm{d}\ln\rho
  =\frac{\bu\!\cdot\!\mathrm{d}\bu}{c^{2}-|\bu|^{2}}
  =\mathrm{d}\ln\gamma(\bu).
  \label{eq:connection-form}
\end{equation}
The frequency connection form is exact.  Consequently continuum transport
along any smooth path multiplies frequencies by
$\gamma(\bu_b)/\gamma(\bu_a)$ independently of the path; every smooth
loop has trivial frequency holonomy; and the discrete two-step holonomy
$\rho=\gamma(\bw)/(\gamma_1\gamma_2)$ of theorem~\ref{thm:rho} is
precisely the $\gamma$-composition defect of finite steps.  Since the
per-step rotation is quadratic in the increment
(proposition~\ref{prop:M4}), the rotation channel also vanishes in the
continuum limit: \emph{the boost sector is continuum-flat}.
\end{proposition}

\begin{proof}
\emph{Step 1.}  Setting $\bu_1=\bu$ and $\bu_2=\mathrm d\bu$ in
\eqref{eq:rho-exact} and expanding to first order in $\mathrm d\bu$
gives the stated form of $q_{00}$, hence
$\mathrm d\ln\rho=\bu\!\cdot\!\mathrm d\bu/(c^{2}-|\bu|^{2})
 =\mathrm d\ln\gamma(\bu)$.

\emph{Step 2.}  The one-form is exact, so its integral along any smooth
path from $\bu_a$ to $\bu_b$ equals
$\ln\bigl[\gamma(\bu_b)/\gamma(\bu_a)\bigr]$, independently of the
path, and vanishes on loops.  The identification of the discrete defect
and the vanishing of the continuum rotation channel follow from
theorem~\ref{thm:rho} and proposition~\ref{prop:M4} respectively.
\end{proof}

\subsection{Lie-algebra structure and the localisation theorem}

\begin{lemma}[generators]\label{lem:generators}
The one-parameter families of canonical maps have the expansions
\begin{equation}
  g(\varepsilon\,\delta\bv)=I+\varepsilon
   \begin{pmatrix}0&0\\-\delta\bv&0\end{pmatrix}+O(\varepsilon^2),
  \qquad
  \ell(-\varepsilon\,\delta\bv)=I+\varepsilon
   \begin{pmatrix}0&\delta\bv\T/c^2\\ \delta\bv&0\end{pmatrix}
   +O(\varepsilon^2),
\end{equation}
so that the boost generator is the strictly upper time-row matrix
\begin{equation}
  b(\delta\bv)\;\equiv\;
  \frac{\mathrm d}{\mathrm d\varepsilon}\Big|_{0} w(\varepsilon\,\delta\bv)
  =\begin{pmatrix}0&\delta\bv\T/c^{2}\\ 0&0\end{pmatrix} :
  \label{eq:bgen}
\end{equation}
\emph{the Galilean factor cancels exactly the spatial--time column of the
Lorentz generator}---the entries responsible for the Thomas--Wigner
rotation of pure boosts.  For an anisotropy increment,
$w_3(\bm 0,\,c^2I+\varepsilon\,\delta C)
 =I+\varepsilon\,a(\delta C)+O(\varepsilon^2)$ with
\begin{equation}
  a(\delta C)=\begin{pmatrix}0&0\\[2pt] 0&-\,\delta C/(2c^{2})
  \end{pmatrix}.
  \label{eq:agen}
\end{equation}
(Item V6a verifies \eqref{eq:bgen} along three directions.)
\end{lemma}

\begin{proof}
\emph{Boost generator.}  The Galilean expansion is exact.  For the
Lorentz factor, $\gamma=1+O(\varepsilon^{2})$, so the time row of
$\ell(-\varepsilon\,\delta\bv)$ is
$(\gamma,\ \gamma\varepsilon\,\delta\bv\T/c^{2})$ and its spatial--time
column is $+\gamma\varepsilon\,\delta\bv$.  Multiplying the two factors
and discarding $O(\varepsilon^{2})$, the spatial--time columns cancel
between the factors and \eqref{eq:bgen} remains.

\emph{Anisotropy generator.}  For $A=c^{2}I+\varepsilon\,\delta C$ at
$\bv=\bm0$, proposition~\ref{prop:w3} gives $p=1$, $\bq=\bm0$ and
$S=c\,A^{-1/2}=I-\varepsilon\,\delta C/(2c^{2})+O(\varepsilon^{2})$,
which is \eqref{eq:agen}.
\end{proof}

\begin{theorem}[localisation of curvature]\label{thm:curvature}
Write $\mathfrak g_f=\mathfrak h\oplus\mathfrak m$ with
$\mathfrak h=\mathfrak{so}(d)$ (spatial rotations, the Lie algebra of the
stabiliser of $\Qstar$, lemma~\ref{lem:stabiliser}) and
$\mathfrak m=\bigl\{\bigl(\begin{smallmatrix}\alpha&\bm\beta\T\\
0&\Sigma\end{smallmatrix}\bigr):\ \Sigma=\Sigma\T\bigr\}$ (the tangent
space of the rotation-free section).  The decomposition is reductive,
$[\mathfrak h,\mathfrak m]\subseteq\mathfrak m$, and the generators
\eqref{eq:bgen}--\eqref{eq:agen} lie in $\mathfrak m$ with brackets
\begin{equation}
  [\,b(\delta\bv_1),\,b(\delta\bv_2)\,]=0,\qquad
  [\,b(\delta\bv),\,a(\delta C)\,]
   = b\!\left(-\frac{\delta C\,\delta\bv}{2c^{2}}\right)
   \in\mathfrak m,
  \label{eq:brackets-bm}
\end{equation}
\begin{equation}
  [\,a(\delta C_1),\,a(\delta C_2)\,]
  =\begin{pmatrix}0&0\\[2pt]
   0&\dfrac{[\delta C_1,\delta C_2]}{4c^{4}}\end{pmatrix}
   \;\in\;\mathfrak h .
  \label{eq:brackets-aa}
\end{equation}
Consequently the canonical connection associated with the section
(curvature $-[X,Y]_{\mathfrak h}$ and torsion $-[X,Y]_{\mathfrak m}$ at
the base point; \citealp{kobayashinomizu1969}) has \emph{rotational
curvature identically zero on boost--boost and boost--anisotropy planes,
and on anisotropy--anisotropy planes equal to the commutator
$[\delta C_1,\delta C_2]/(4c^4)$}, nonzero exactly when the two
sound-speed-tensor increments fail to commute (e.g.\ a normal stretch and
a $45^\circ$ shear).  The connection carries torsion on mixed
boost--anisotropy planes, by the second bracket of
\eqref{eq:brackets-bm}.
\end{theorem}

\begin{proof}
\emph{Step 1: reductivity.}  For $\Omega\in\mathfrak{so}(d)$ embedded as
$\bigl(\begin{smallmatrix}0&0\\0&\Omega\end{smallmatrix}\bigr)$ and
$X=\bigl(\begin{smallmatrix}\alpha&\bm\beta\T\\0&\Sigma
\end{smallmatrix}\bigr)\in\mathfrak m$,
\begin{equation}
  \Bigl[\begin{pmatrix}0&0\\0&\Omega\end{pmatrix},\,X\Bigr]
  =\begin{pmatrix}0&-\bm\beta\T\Omega\\[2pt]
    0&\Omega\Sigma-\Sigma\Omega\end{pmatrix},
\end{equation}
and $\Omega\Sigma-\Sigma\Omega$ is symmetric, so
$[\mathfrak h,\mathfrak m]\subseteq\mathfrak m$.

\emph{Step 2: brackets.}  With
$b=\bigl(\begin{smallmatrix}0&\bm\beta\T\\0&0\end{smallmatrix}\bigr)$,
$\bm\beta=\delta\bv/c^{2}$, and
$a=\bigl(\begin{smallmatrix}0&0\\0&M\end{smallmatrix}\bigr)$,
$M=-\delta C/(2c^{2})$:
(i)~the product of two $b$-type matrices vanishes identically, so
$[b_1,b_2]=0$;
(ii)~$b\,a=\bigl(\begin{smallmatrix}0&\bm\beta\T M\\0&0
\end{smallmatrix}\bigr)$ while $a\,b=0$, so
$[b,a]=b(c^{2}M\bm\beta)=b(-\delta C\,\delta\bv/2c^{2})\in\mathfrak m$;
(iii)~$[a_1,a_2]$ has spatial block
$[M_1,M_2]=[\delta C_1,\delta C_2]/(4c^{4})$, which is antisymmetric
(the commutator of symmetric matrices), i.e.\ a rotation generator in
$\mathfrak h$.  (Items V6b--V6d.)

\emph{Step 3: conclusion.}  For the canonical connection of a reductive
homogeneous space, the curvature and torsion at the base point are
$-[X,Y]_{\mathfrak h}$ and $-[X,Y]_{\mathfrak m}$ respectively
\citep{kobayashinomizu1969}; the statements follow by reading off the
$\mathfrak h$- and $\mathfrak m$-components of
\eqref{eq:brackets-bm}--\eqref{eq:brackets-aa}.
\end{proof}

\begin{proposition}[discrete check: the stretch--shear triangle]
\label{prop:bch}
For finite anisotropy steps $\delta C_1=\varepsilon_1c^2M_1$,
$\delta C_2=\varepsilon_2c^2M_2$ with $M_1=\operatorname{diag}(1,-1)$ and
$M_2=\bigl(\begin{smallmatrix}0&1\\1&0\end{smallmatrix}\bigr)$, the
two-step-plus-canonical-closure holonomy is a pure rotation
($\rho=1$ exactly, the time row being untouched by spatial maps) of angle
\begin{equation}
  \theta=\frac{\varepsilon_1\varepsilon_2}{4}+O(\varepsilon^3).
  \label{eq:bch-angle}
\end{equation}
\end{proposition}

\begin{proof}
\emph{Step 1.}  By the Baker--Campbell--Hausdorff expansion with
$A_i=-\varepsilon_iM_i/2$,
\begin{equation}
  \e^{A_2}\e^{A_1}
  =\exp\Bigl(A_1+A_2+\tfrac12[A_2,A_1]+\cdots\Bigr),
  \qquad
  \tfrac12[A_2,A_1]
  =\frac{\varepsilon_1\varepsilon_2}{8}\,[M_2,M_1]
  =\frac{\varepsilon_1\varepsilon_2}{8}
   \begin{pmatrix}0&-2\\2&0\end{pmatrix}.
\end{equation}

\emph{Step 2.}  This antisymmetric term is the $\mathfrak h$-part of
the exponent, so by lemma~\ref{lem:polar} the polar rotation angle is
$\theta=\varepsilon_1\varepsilon_2/4+O(\varepsilon^{3})$.  Numerically
(item V7b), $\theta/(\varepsilon_1\varepsilon_2)$ evaluates to
$0.25001$, $0.25003$, $0.25006$ for
$(\varepsilon_1,\varepsilon_2)=(10^{-2},10^{-2})$,
$(2\times10^{-2},10^{-2})$, $(10^{-2},3\times10^{-2})$, converging to
$1/4$ as $\varepsilon\to0$.
\end{proof}

\begin{remark}[what is standard and what is new]\label{rem:spd}
The pure-anisotropy sector ($\bv=\bm0$, conformal classes of
$C\succ0$) is the classical symmetric space of positive-definite
symmetric matrices, $SL(d,\mathbb R)/SO(d)$ up to scale
\citep{helgason1978,bhatia2007}, and \eqref{eq:brackets-aa} is its canonical
curvature.  The acoustically new content is the pair of exact flattenings
around it: pure Lorentz boosts would contribute Thomas--Wigner curvature
at second order, but the Galilean factor in $w(\bv)$ cancels the
responsible generator entries exactly (lemma~\ref{lem:generators}),
demoting all boost-sector rotation to the fourth-order discretisation
defect of proposition~\ref{prop:M4}; and the frequency channel, absent
from the Riemannian picture, is governed by the exact form
\eqref{eq:connection-form}.  The torsion on mixed planes has, at present,
no physical interpretation; we flag it as an open structural question
rather than suppress it.
\end{remark}

\section{Transport of statistical source models and its observable
consequences}\label{sec:stat}

The geometry of \S\S\ref{sec:transport}--\ref{sec:curvature} becomes
falsifiable when it acts on the objects practitioners actually exchange:
statistical source models.  In hybrid prediction methods the source
statistics are modelled or extracted from an unsteady computation
\citep{colonius1997,freund2001,bogey2002,tamauriault1999} and propagated
by the operator of a chosen analogy, a practice reviewed by
\citet{coloniuslele2004} and \citet{wangfreundlele2006}; filtering and
splitting formulations \citep{hardin1994,ewert2003} interpose further
changes of description; and wave-packet source models
\citep{jordancolonius2013,cavalieri2019} are specified directly in the
wavenumber--frequency form on which the transport below acts.
Each change of description is a re-gauging step in the sense of
definition~\ref{def:transport}.  This section computes exactly what a
closed loop of such steps does to a model.

\subsection{What is, and is not, observable}\label{ssec:observable}

Because the results that follow attach a nonzero frequency dilation and a
nonzero rotation to closed loops in analogy space, it is essential to
state at the outset what these quantities are, and what they are not.  The
logic has three steps.

\emph{On shell, nothing changes.}  By proposition~\ref{prop:trivial},
every exact analogy evaluated on a true solution of \eqref{eq:NS} returns
the identical physical far field.  No transport, no loop, and no holonomy
alters a single measurable quantity of the radiated sound.  The geometry
of this paper is therefore \emph{not} a claim that changing analogy
changes the physics; any reading of the dilation $\rho$ as a physical
frequency shift would contradict proposition~\ref{prop:trivial} and is
incorrect.

\emph{Off shell, models are transported.}  An analogy is used by
inserting a \emph{modelled} source---a statistical spectrum, a wave-packet
ansatz, a filtered simulation field---into the solution operator of its
$L$.  A model is not a true solution: it does not satisfy \eqref{eq:NS},
the continuity residual of remark~\ref{rem:kinematic} does not vanish, and
the far field it produces is an approximation whose error depends on the
operator it was fed to.  When the \emph{same} model is reused across
analogies---as it is whenever a spectrum calibrated in one frame is
propagated in another, or a wave-packet model specified in
wavenumber--frequency variables is carried into a convected or generalised
operator---the two pipelines produce genuinely different far fields, both
of them approximate.

\emph{The holonomy measures the discrepancy between two wrong answers.}
The dilation $\rho$ and rotation $R$ of theorem~\ref{thm:structure} are
the exact, closed-form, parameter-free measure of how differently two such
pipelines err when they transport a common model around a loop.  This is
an inter-method bias, not a property of the flow.  It is observable in the
only sense that matters for practice: two groups modelling the same
turbulence with the same source model but different analogies will predict
peak frequencies differing by exactly $\rho$ and directivity patterns
rotated by exactly $R$, with no freedom to reconcile them short of
recalibration.  The result is thus a statement about the
\emph{consistency of hybrid pipelines}, and it is precisely the
transported-model discrepancy---not any change in the physics---that
\S\ref{ssec:caveats} quantifies and figures~\ref{fig:rho}--\ref{fig:spectra}
display.

\subsection{Why this transport, and not another}\label{ssec:whytransport}

A geometric defect is only meaningful if the connection generating it is
canonical rather than chosen for convenience; otherwise a referee may
reasonably ask why \emph{this} holonomy, rather than one obtained from
some other transport rule, deserves attention.  The transport of
definition~\ref{def:transport} is not one option among many: it is forced,
at each stage, by the way prediction is actually carried out.

\emph{Frequency-respect is forced by fixed-frequency prediction.}
Green's functions are computed, source spectra specified, and measurements
compared at fixed frequency.  A transport that did not preserve the
time-harmonic decomposition would smear a single-frequency model across
frequencies and could not be composed with the frequency-by-frequency
machinery of prediction.  By lemma~\ref{lem:block}, respecting that
decomposition forces the block form \eqref{eq:blockform} exactly---no
larger class is admissible.

\emph{The reduction target and the rotation-free normalisation are
forced.}  Every analogy is defined by reduction of its operator to the
common d'Alembertian; the maps effecting that reduction are fixed up to
the stabiliser of $\Qstar$, which by lemma~\ref{lem:stabiliser} is exactly
$\mathbb{R}_{>0}\times O(d)$---a frequency scaling and a spatial rotation,
no more.  Requiring the canonical representative to carry no spurious
rotation (proposition~\ref{prop:w3}, the symmetric positive-definite
section) removes the $O(d)$ freedom uniquely.  What remains, the residual
$\mathbb{R}_{>0}$ scaling and the rotation that survives around a loop, is
therefore not an artefact of a particular gauge choice but the intrinsic
holonomy of the unique frequency-respecting, rotation-free connection.
The dilation and rotation are canonical given only the premise that
aeroacoustic prediction is performed at fixed frequency.

\subsection{Transport of sources and of their statistics}

\begin{lemma}[source transport]\label{lem:source}
Let $\Box p=s$ in the quiescent analogy and let $D=\rho\,(1\oplus R)$ be
a holonomy \eqref{eq:Dform}.  Then, with $\xi'=D\xi$,
\begin{equation}
  p'\;=\;p\circ D^{-1},
  \qquad
  s'\;=\;\rho^{-2}\,s\circ D^{-1}
  \qquad\Longrightarrow\qquad
  \Box p'\;=\;s' .
  \label{eq:source-transport}
\end{equation}
\end{lemma}

\begin{proof}
Apply \eqref{eq:congruence} with $M=D$ and
$D\Qstar D\T=\rho^{2}\Qstar$:
\begin{equation}
  s(\xi)=(P_{\Qstar}p)(\xi)=(P_{\rho^{2}\Qstar}p')(\xi')
        =\rho^{2}\,(\Box p')(\xi'),
\end{equation}
i.e.\ $(\Box p')(\xi')=\rho^{-2}s(D^{-1}\xi')=s'(\xi')$.
\end{proof}

\begin{lemma}[transport of the cross-spectral model]\label{lem:csd}
Let $s$ be a stationary random source with correlation
$R_s(\zeta)=\langle s(\xi)\,s(\xi+\zeta)\rangle$ and
wavenumber--frequency spectrum
$\Phi(\omega,\bk)=\int R_s(\zeta)\,\e^{-\im\kappa\cdot\zeta}\dd\zeta$
(pairing $\kappa\cdot\zeta=\bk\!\cdot\!\bx-\omega t$).  Under
\eqref{eq:source-transport}, stationarity is preserved and
\begin{equation}
  \Phi'(\kappa)
  =\rho^{-4}\,\lvert\det D\rvert\;\Phi\bigl(D\T\kappa\bigr)
  \;\overset{d=3}{=}\;
  \Phi\bigl(\rho\,\omega,\ \rho\,R\T\bk\bigr),
  \label{eq:csd-transport}
\end{equation}
since $\lvert\det D\rvert=\rho^{\,1+d}=\rho^{4}$ in three space
dimensions: \emph{the amplitude weights cancel identically, and the model
spectrum is transported by rigid dilation and rotation of its arguments.}
(In $d=2$ a residual amplitude factor $\rho^{-1}$ survives.)
\end{lemma}

\begin{proof}
\emph{Step 1.}  From \eqref{eq:source-transport},
$R_{s'}(\zeta)=\rho^{-4}R_s(D^{-1}\zeta)$; stationarity is preserved
because $D$ is linear.

\emph{Step 2.}  Substituting $\zeta=D\eta$ in the Fourier integral,
\begin{equation}
  \Phi'(\kappa)
  =\rho^{-4}\lvert\det D\rvert\int R_s(\eta)\,
   \e^{-\im(D\T\kappa)\cdot\eta}\dd\eta
  =\rho^{-4}\lvert\det D\rvert\;\Phi\bigl(D\T\kappa\bigr).
\end{equation}

\emph{Step 3.}  For $D=\rho(1\oplus R)$,
$\lvert\det D\rvert=\rho^{1+d}$ and $D\T$ acts on $(\omega,\bk)$ as
$(\omega,\bk)\mapsto\rho(\omega,R\T\bk)$, the time and space components
not mixing; for $d=3$ the weights cancel, giving
\eqref{eq:csd-transport}.
\end{proof}

\subsection{The two-pipeline theorem}

For a compact stationary source in the quiescent medium the far-field
pressure spectral density samples the source spectrum on the sonic
sphere: with the free-space Green's function
$\e^{\im\omega r/c}/4\pi r$ \citep{goldstein1976,morseingard1968},
\begin{equation}
  S_{pp}(r\bn,\omega)\;=\;\frac{K}{r^{2}}\,
  \Phi\Bigl(\omega,\ \frac{\omega}{c}\,\bn\Bigr)
  \qquad (r\to\infty),
  \label{eq:farfield-sampling}
\end{equation}
with $K$ a positive convention constant.  Define the directivity
$I(\bn,\omega)\equiv\lim_{r\to\infty}r^{2}S_{pp}(r\bn,\omega)$.

\begin{theorem}[two-pipeline law]\label{thm:pipelines}
Let two prediction pipelines carry the same source-level model to the
quiescent analogy along paths whose relative holonomy is
$D=\rho(1\oplus R)$.  Then their far-field predictions are related
exactly ($d=3$) by
\begin{equation}
  I_2(\bn,\omega)\;=\;I_1\bigl(R\T\bn,\ \rho\,\omega\bigr):
  \label{eq:pipelines}
\end{equation}
a rigid rotation of the directivity pattern by $R$ together with a rigid
dilation of the spectrum by $\rho$ (in particular, predicted spectral
peaks differ by exactly the factor $\rho$: the transported spectrum
$\Phi(\rho\omega,\cdot)$ peaks at $\omega_0/\rho$).
\end{theorem}

\begin{proof}
Insert \eqref{eq:csd-transport} into \eqref{eq:farfield-sampling}:
\begin{equation}
  I_2(\bn,\omega)
  =K\,\Phi\Bigl(\rho\omega,\ \rho\,\frac{\omega}{c}\,R\T\bn\Bigr)
  =K\,\Phi\Bigl(\rho\omega,\ \frac{\rho\omega}{c}\,R\T\bn\Bigr)
  =I_1\bigl(R\T\bn,\rho\omega\bigr).
\end{equation}
The transported radiating wavevector lies exactly on the sonic sphere of
the dilated frequency: the transport is sonic-sphere covariant, with no
leakage between radiating and non-radiating content.
\end{proof}

\begin{remark}[amplitude conventions]
If models are specified at a base-field level, $s=P_q(\partial)b$ with
$P_q$ homogeneous of degree $q$ (e.g.\ $q=2$ for
$\partial_i\partial_jT_{ij}$), transporting the base model instead of the
source model changes \eqref{eq:pipelines} by the factor $\rho^{2q}$,
since $|\hat P_q(D\T\kappa)|^{2}=\rho^{2q}\,
|\hat P_q((1\oplus R\T)\kappa)|^{2}$.  The rigid rotation--dilation
structure is convention-independent; only the overall amplitude weight
moves.
\end{remark}

\begin{corollary}[exact complementarity of the two channels]
\label{cor:complementarity}
\emph{(i) Anisotropy loops.}  For loops composed of anisotropy steps at
$\bv=\bm0$, every elementary map has time row $(1,\bm0)$, so
$\rho=1$ \emph{exactly}: the loop acts as a pure rigid rotation of the
directivity pattern, leaving spectra untouched.
\emph{(ii) Boost loops.}  For loops composed of boost steps,
$R=I+O(M^{4})$ (proposition~\ref{prop:M4}): to fourth order the loop acts
as a pure rigid spectral dilation, leaving directivity untouched.
The two observable channels of the holonomy are carried by the two
sectors of the geometry, matching exactly the localisation of curvature
(theorem~\ref{thm:curvature}).
\end{corollary}

\subsection{Numerical verification and magnitudes}

Table~\ref{tab:verification} collects the end-to-end numerical checks;
figures~\ref{fig:rho}--\ref{fig:spectra} display the corresponding maps.
All computations use the full pipeline of
definition~\ref{def:transport} and proposition~\ref{prop:w3} (compose
elementary maps, extract the endpoint medium, construct the canonical
closure, invert), independent of the closed-form derivations they are
compared with.

\begin{table}
\centering
\begin{tabular}{llll}
\hline
Check & Computed & Predicted & Item\\
\hline
Boost loop $M_1{=}0.3$, $M_2{=}0.4$: $\rho$ &
 $1.153096610988149$ & $1.153096610988149$ & V7a\\
\quad rotation residual $|D-\rho I|_{\max}$ &
 $2.2\times10^{-16}$ & $0$ (collinear) & V7a\\
\quad transported spectral peak &
 $0.867233\,\omega_0$ & $\omega_0/\rho=0.867230\,\omega_0$ & \S\ref{sec:stat}\\
Anisotropy loop $\varepsilon_1{=}\varepsilon_2{=}0.05$: $\rho$ &
 $1.000000000000$ & $1$ (exact) & V7b\\
\quad rotation $\theta$ (rad) &
 $6.2578\times10^{-4}$ & $\varepsilon_1\varepsilon_2/4=6.25\times10^{-4}$
 & V7b\\
\quad quadrupole lobe shift (rad) &
 $6.283\times10^{-4}$ & $\theta$ (grid $7.9\times10^{-6}$) & V7c\\
\quad lobe-shape distortion &
 $2.6\times10^{-11}$ & $0$ (rigid rotation) & \S\ref{sec:stat}\\
Stretch--shear ratio $\theta/(\varepsilon_1\varepsilon_2)$ &
 $0.25001$--$0.25006$ & $1/4$ & V7b\\
\hline
\end{tabular}
\caption{End-to-end numerical verification of the transport laws.  Item
labels refer to the verification record of
appendix~\ref{app:verification}.}
\label{tab:verification}
\end{table}

\begin{figure}
\centering
\includegraphics[width=0.62\textwidth]{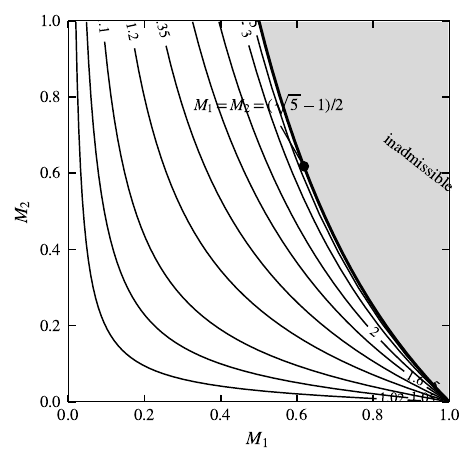}
\caption{Exact frequency holonomy $\rho$ of two collinear re-gauging
steps $(M_1,M_2)$, equation \eqref{eq:rho-exact} (contours), the sonic
horizon $M_1M_2+M_1=1$ of theorem~\ref{thm:horizon} (heavy curve,
$\rho\to\infty$), the inadmissible region beyond it (shaded), and the
equal-step degeneration point $M_1=M_2=(\sqrt5-1)/2$
\eqref{eq:golden}.  Note the asymmetry in $(M_1,M_2)$: transport is
path-ordered (corollary~\ref{cor:einstein}).}
\label{fig:rho}
\end{figure}

\begin{figure}
\centering
\includegraphics[width=\textwidth]{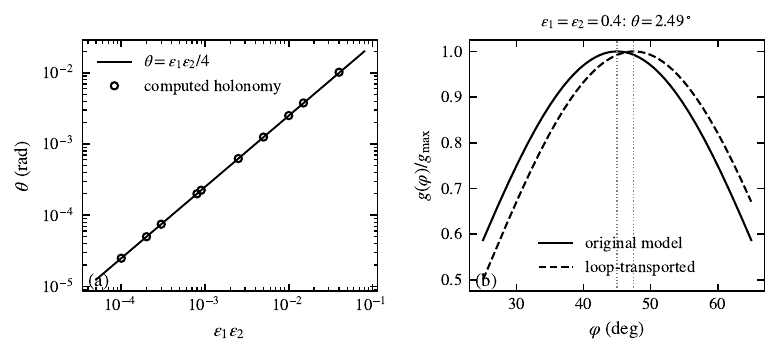}
\caption{Rotation channel.  (a)~Computed anisotropy-loop holonomy angle
against $\varepsilon_1\varepsilon_2$ for stretch--shear step pairs,
compared with the curvature prediction
$\theta=\varepsilon_1\varepsilon_2/4$
(proposition~\ref{prop:bch}); departures at the largest steps are the
expected $O(\varepsilon^3)$ corrections.  (b)~Directivity of an
$x_1x_2$-quadrupole model before and after loop transport at
$\varepsilon_1=\varepsilon_2=0.4$: the lobe rotates rigidly by
$\theta=2.49^\circ$ (leading order $2.29^\circ$) with shape preserved to
better than $10^{-10}$ (table~\ref{tab:verification}); vertical lines
mark the lobe maxima.}
\label{fig:theta}
\end{figure}

\begin{figure}
\centering
\includegraphics[width=0.55\textwidth]{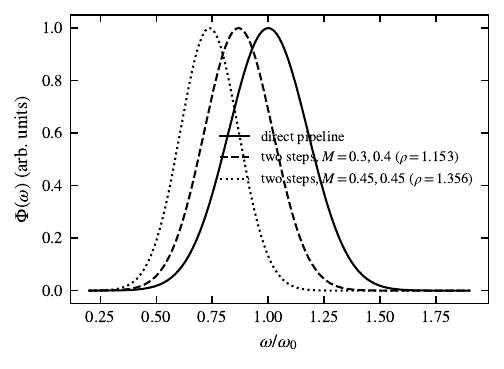}
\caption{Frequency channel: a model source spectrum (solid) and its
loop transports for two two-step collinear pipelines,
equation~\eqref{eq:pipelines}.  The dilation is rigid; predicted peak
frequencies differ from the direct pipeline by the exact factors
$\rho=1.153$ ($M=0.3,0.4$) and $\rho=1.356$ ($M=0.45,0.45$).}
\label{fig:spectra}
\end{figure}

\emph{Magnitudes.}  For small steps,
$\rho-1\simeq M_1M_2\cos\varphi$ by theorem~\ref{thm:rho}: two
$M=0.3$ steps give a $9\%$ systematic spectral misalignment.  The
worst-case envelope is large: two equal collinear steps at $M=0.45$ give
$\rho=1.356$, a $36\%$ ($\approx0.44$-octave) displacement of every
spectral feature (figure~\ref{fig:spectra}), and $\rho$ diverges at the
horizon.  These are not modelling errors: they are exact, systematic,
computable properties of the re-gauging path, and \eqref{eq:pipelines}
predicts them with no free parameters.  For the rotation channel,
$\theta\simeq\varepsilon_1\varepsilon_2/4$ per stretch--shear loop of
relative medium anisotropy $\varepsilon$; at $\varepsilon\sim0.2$,
$\theta\sim0.6^{\circ}$.  We stress that the mapping from modelled
Reynolds-stress anisotropy to effective-medium anisotropy $\delta C$ is
an open modelling link (remark~\ref{rem:goldstein}), so the rotation
magnitudes are parametric where the dilation magnitudes are not.

\subsection{The far field of a generalised medium and open paths}
\label{ssec:anisofar}

Theorem~\ref{thm:pipelines} assumed loops based at the quiescent analogy.
Removing that restriction requires the far field radiated \emph{in} an
admissible generalised medium, where group-velocity and wavevector
directions part ways.  The canonical map supplies it in closed form;
remarkably, the data $(A,p,\bq)$ of proposition~\ref{prop:w3} turn out
to be precisely the geometric data of the radiation problem.

\begin{lemma}[far field of a generalised medium]\label{lem:anisofar}
Let $(\bv,C)$ be admissible, $A=C-\bv\bv\T$, and let $p_\omega$ be the
outgoing (limiting-absorption) solution of
$P_{Q(\bv,C)}\bigl(p_\omega\,\e^{-\im\omega t}\bigr)
 =s\,\e^{-\im\omega t}$ with $s$ compactly supported ($d=3$).  Then, as
$r\to\infty$ along fixed $\bn$,
\begin{equation}
  p_\omega(r\bn)
  =\frac{\hat s\bigl(\bk_*(\bn,\omega)\bigr)\,
         \e^{\,\im(\bk_*\cdot\,\bn)\,r}}
        {4\pi r\,\sqrt{\det A}\;\sqrt{\bn\T A^{-1}\bn}}
   \;+\;O(r^{-2}),
  \label{eq:anisofar}
\end{equation}
where, with $\boldsymbol{\nu}(\bn)=\bn/\sqrt{\bn\T A^{-1}\bn}$ and $p$
as in proposition~\ref{prop:w3},
\begin{equation}
  \bk_*(\bn,\omega)=\omega\,A^{-1}
   \Bigl(\frac{\boldsymbol{\nu}}{p}-\bv\Bigr),
  \qquad
  \bk_*\!\cdot\bn
  =\omega\Bigl(\frac{\sqrt{\bn\T A^{-1}\bn}}{p}-\bv\T A^{-1}\bn\Bigr).
  \label{eq:kstar}
\end{equation}
Moreover $\bk_*$ is characterised intrinsically: it is the unique
wavevector on the outgoing sonic ellipsoid
\begin{equation}
  E_\omega:\quad \bk\T A\,\bk+2\omega\,\bv\!\cdot\!\bk=\omega^{2},
  \qquad\text{equivalently}\quad
  \omega-\bv\!\cdot\!\bk=+\sqrt{\bk\T C\,\bk},
  \label{eq:ellipsoid}
\end{equation}
whose group velocity
$\bv_g(\bk)=\bv+C\bk/\sqrt{\bk\T C\bk}$ is positively parallel to
$\bn$; the group speed is
$|\bv_g|=\bigl[\sqrt{\bn\T A^{-1}\bn}\,\mu\bigr]^{-1}$ with
$\mu=1/p-\bv\T A^{-1}\boldsymbol{\nu}>0$.
\end{lemma}

\begin{proof}
\emph{Step 1: exact reduction of the Fourier representation.}  Write
the outgoing solution as
\begin{equation}
  p_\omega(\bx)=\int
  \frac{\hat s(\bk)\,\e^{\im\bk\cdot\bx}}{h(\bk,\omega)}
  \,\frac{\dd^3\bk}{(2\pi)^3},
  \qquad
  h(\bk,\omega)\equiv\bk\T C\bk-(\omega-\bv\!\cdot\!\bk)^{2},
  \qquad \omega\to\omega+\im0 .
  \label{eq:fourier-rep}
\end{equation}
Under the exact substitution $\bk=S\T\bk'-\omega'\bq$,
$\omega=p\,\omega'$ (the covector action of $w_3$), the symbol is
carried to the plain one,
\begin{equation}
  h(\bk,\omega)\;=\;c^{2}|\bk'|^{2}-\omega'^{2},
  \label{eq:symbol-congruence}
\end{equation}
which is the congruence
$\kappa\T Q\kappa=\kappa'^{\mathrm T}\Qstar\kappa'$ for
$\kappa=w_3\T\kappa'$ (proposition~\ref{prop:w3}; verified in general
symbolic form, item~V8a); the Jacobian is $\lvert\det S\rvert$ and,
because $p>0$, $\omega+\im0$ is carried to $\omega'+\im0$: \emph{the
radiation condition transports}.  (This is where admissibility enters:
at the horizon of theorem~\ref{thm:horizon} no such transport exists.)
Consequently
\begin{equation}
  p_\omega(\bx)=\e^{-\im\omega'\bq\cdot\bx}\,p'_{\omega'}(S\bx),
  \qquad
  \hat s'(\bk')=\lvert\det S\rvert\,
   \hat s\bigl(S\T\bk'-\omega'\bq\bigr),
  \label{eq:field-correspondence}
\end{equation}
where $p'$ is the outgoing plain-medium solution with source spectrum
$\hat s'$.

\emph{Step 2: isotropic far field.}
$p'_{\omega'}(r'\bn')=\e^{\im\omega'r'/c}\,
 \hat s'\bigl((\omega'/c)\bn'\bigr)/(4\pi c^{2}r')+O(r'^{-2})$
\citep{goldstein1976}.

\emph{Step 3: pullback.}  Along $\bx=r\bn$ we have
$S\bx=r\,|S\bn|\,\bn'$ with $\bn'=S\bn/|S\bn|$; assembling the phases
of \eqref{eq:field-correspondence},
\begin{equation}
  -\omega'\bq\cdot\bn+\frac{\omega'}{c}\,|S\bn|
  =\bk_*\!\cdot\bn,
  \qquad
  \bk_*=\frac{\omega'}{c}\,S\T\bn'-\omega'\bq,
\end{equation}
and $S\T S=c^{2}A^{-1}$, $\bq=pA^{-1}\bv$ reduce $\bk_*$ to
\eqref{eq:kstar} and the amplitude to
$\lvert\det S\rvert/(4\pi c^{2}|S\bn|\,r)
 =\bigl[4\pi r\sqrt{\det A\;\bn\T A^{-1}\bn}\,\bigr]^{-1}$.

\emph{Step 4: intrinsic identification.}  All checks are one-line
computations in the variables $A^{-1}$ (verified symbolically for
generic three-dimensional media, items V8b).  \emph{On shell:}
\begin{equation}
  \bk_*\T A\bk_*+2\omega\,\bv\!\cdot\!\bk_*
  =\omega^{2}\Bigl(\frac{\boldsymbol{\nu}\T A^{-1}\boldsymbol{\nu}}
   {p^{2}}-\bv\T A^{-1}\bv\Bigr)=\omega^{2},
\end{equation}
using $\boldsymbol{\nu}\T A^{-1}\boldsymbol{\nu}=1$ and
$1/p^{2}=1+\bv\T A^{-1}\bv$.  \emph{Branch:}
$\omega-\bv\!\cdot\!\bk_*=(\omega/p)\,\mu>0$ since
$\bv\T A^{-1}\boldsymbol{\nu}\le\sqrt{\bv\T A^{-1}\bv}<1/p$
(Cauchy--Schwarz in the $A^{-1}$ inner product).  \emph{Ray direction:}
for every $\bk$,
\begin{equation}
  C\bk+(\omega-\bv\!\cdot\!\bk)\,\bv
  \;=\;A\bk+\omega\bv
  \qquad(\text{since }C-A=\bv\bv\T),
\end{equation}
so $\bv_g$ is parallel to the outward normal
$\tfrac12\nabla_{\bk}\bigl(\bk\T A\bk+2\omega\bv\!\cdot\!\bk\bigr)$ of
$E_\omega$; at $\bk_*$ this normal is
$A\bk_*+\omega\bv=\omega\boldsymbol{\nu}/p\parallel\bn$, and since
$E_\omega$ is an ellipsoid ($A\succ0$) its Gauss map is a bijection,
$\bk_*$ is unique.  \emph{Group speed:} from
$C\bk_*=\omega(\boldsymbol{\nu}/p-\bv)+(\bv\!\cdot\!\bk_*)\bv$ and
$\sqrt{\bk_*\T C\bk_*}=(\omega/p)\mu$.

\emph{Step 5: consistency and numerical validation.}  For $A=c^2I$,
$\bv=\bm0$, \eqref{eq:anisofar} is the standard formula.  For the
classical convected medium ($C=c^2I$, $\bv=U\bm e_1$),
$\det A=c^{6}\beta^{2}$ and the amplitude reduces to
$\bigl[4\pi c^{2}r\sqrt{n_1^{2}+\beta^{2}|\bn_\perp|^{2}}\,\bigr]^{-1}$,
the classical uniform-flow monopole directivity with Prandtl--Glauert
radius $R_\beta=r(n_1^2+\beta^2|\bn_\perp|^2)^{1/2}$.  An end-to-end
numerical validation against a direct limiting-absorption FFT solution
of the anisotropic convected Helmholtz equation---a computation that
performs no change of variables---confirms \eqref{eq:anisofar} in
amplitude, phase and $\bn$-dependence (item~V8c).
\end{proof}

\begin{corollary}[open-path pipeline law]\label{cor:openpath}
Let two pipelines carry the same source-level model along different
admissible paths to a common endpoint $(\bv,C)$, with relative holonomy
$D=\rho(1\oplus R)$ in the reduced frame, and let
$\mathfrak r(\bn)=S^{-1}R\T S\bn/\lvert S^{-1}R\T S\bn\rvert$ denote the
conjugated direction map.  Then their far-field predictions in the
endpoint medium satisfy, exactly ($d=3$),
\begin{equation}
  I_2(\bn,\omega)
  =\frac{\lvert S\,\mathfrak r(\bn)\rvert^{2}}{\lvert S\bn\rvert^{2}}\;
   I_1\bigl(\mathfrak r(\bn),\ \rho\,\omega\bigr).
  \label{eq:openpath}
\end{equation}
\end{corollary}

\begin{proof}
\emph{Step 1.}  By lemma~\ref{lem:anisofar} in statistical form (the
sampling structure of \S\ref{sec:stat} carries over verbatim), each
pipeline predicts
\begin{equation}
  I_i(\bn,\omega)=K(\bn)\,
  \Phi'_i\Bigl(\frac{\omega}{p},\ \frac{\omega}{pc}\,
  \frac{S\bn}{|S\bn|}\Bigr),
  \qquad K(\bn)\propto|S\bn|^{-2},
\end{equation}
where the direction-dependent factor $K(\bn)$ collects the
canonical-map Jacobians and the fixed-observer time dilation, is
identical for both pipelines, and reduces to the constant $K$ of
\eqref{eq:farfield-sampling} at an isotropic endpoint.

\emph{Step 2.}  By lemma~\ref{lem:csd}, which applies verbatim,
$\Phi'_2(\kappa')=\Phi'_1(D\T\kappa')$, so
\begin{equation}
  I_2(\bn,\omega)=K(\bn)\,
  \Phi'_1\Bigl(\frac{\rho\omega}{p},\ \frac{\rho\omega}{pc}\,
  \frac{R\T S\bn}{|S\bn|}\Bigr).
\end{equation}

\emph{Step 3.}  The direction $\bm m=\mathfrak r(\bn)$ satisfies
$S\bm m/|S\bm m|=R\T S\bn/|S\bn|$ ($R$ being orthogonal), whence the
right-hand side equals
$[K(\bn)/K(\bm m)]\,I_1(\bm m,\rho\omega)$, and
$K(\bn)/K(\bm m)=|S\bm m|^{2}/|S\bn|^{2}$, which is
\eqref{eq:openpath}.
\end{proof}

\begin{remark}
On the reduced sphere the discrepancy between the pipelines is the rigid
rotation $R$; on the physical direction sphere it is the conjugated,
generally non-isometric map $\mathfrak r$---patterns are rotated rigidly
in the geometry of the sonic ellipsoid, not of the observation
sphere---and the prefactor in \eqref{eq:openpath} is exactly the
solid-angle Jacobian that conservation of radiated power along ray
tubes requires.  At an isotropic endpoint ($S\propto I$) the corollary
reduces to theorem~\ref{thm:pipelines}.
\end{remark}

\subsection{Scope and caveats}\label{ssec:caveats}

One restriction of theorem~\ref{thm:pipelines} has been removed and two
limitations remain.
(i)~\emph{Anisotropic endpoints.}  As stated, the theorem concerns loops
based at the quiescent analogy, where \eqref{eq:farfield-sampling}
applies at both ends and all anisotropic bookkeeping is interior.
Lemma~\ref{lem:anisofar} removes the restriction for scalar sources:
the far field of a generalised medium samples the source spectrum on
the sonic ellipsoid along group rays---where group-velocity and
wavevector directions differ (cf.\ the disturbance-energy transport of
\citealp{myers1991})---in exactly the form transported by the canonical
map, and corollary~\ref{cor:openpath} extends the pipeline law to open
paths.
(ii)~\emph{Scalar sources.}  Lighthill-type tensor sources transport
with an additional $R\otimes R$ action on indices; the rigid-rotation
structure survives, but the polarisation bookkeeping is deferred.
(iii)~\emph{Dimension.}  The clean cancellation in
\eqref{eq:csd-transport} is specific to $d=3$; planar problems carry a
residual $\rho^{-1}$ amplitude weight.

\section{Discussion}\label{sec:discussion}

\subsection{What the geometry says about the ``true source'' debate}

The debate over the true sources of aerodynamic sound
\citep{ffowcswilliams1982,tam1998,goldstein2005} has always contained two
intertwined claims: that source attribution is a matter of choice, and
that the choice matters in practice.  The present framework separates and
quantifies both.  At the exact level the attribution question is
\emph{empty} (proposition~\ref{prop:trivial}), and re-attribution is
kinematic bookkeeping (lemma~\ref{lem:kinematic})---this is the rigorous
core of the ``fiction'' half of \citet{ffowcswilliams1982}.  At the level
of use, however, the ambiguity is not amorphous.  Within the uniform
sector the space of analogies is \emph{flat}
(proposition~\ref{prop:flat}): a continuum notion of ``the same model in
a different analogy'' exists canonically, and everything that finite
re-gauging pipelines disagree about is the computable discretisation
defect $(\rho,R)$ of theorems~\ref{thm:rho}
and~\ref{thm:structure}---including its divergence at the sonic horizon
(theorem~\ref{thm:horizon}).  Across anisotropy directions, by contrast,
there is an irreducible curvature obstruction
(theorem~\ref{thm:curvature}): no path-independent identification of
source models exists at all, and the ambiguity has a definite tensorial
shape, the commutator of medium deformations.  We suggest this is the
precise mathematical content of the intuition, expressed differently by
\citet{tam1998} and \citet{goldstein2005}, that the search for a unique
physical source is ill-posed while the choice of analogy remains
consequential: the choice lives on a space with flat directions, where it
is gauge in the strict sense, and curved directions, where it is not.

\subsection{Relation to the generalised analogy and to non-radiating
source theory}

Theorem~\ref{thm:nonclosure} gives the generalised acoustic analogy of
\citet{goldstein2003} a structural, rather than merely practical,
justification: anisotropic effective media are forced by closure of
re-gauging, independently of any modelling considerations.  The known
delicacy of parallel-flow analogies \citep{pridmorebrown1958,
goldstein2001,musafir2007}---instability of Lilley-type operators and
the sensitivity of their source decompositions---sits naturally in this
picture as behaviour near degenerate loci of analogy space, of which the
horizon of theorem~\ref{thm:horizon} is the exactly computable uniform
prototype.  The documented sensitivity of educed source fields to the
choice of analogy and to input errors \citep{samanta2006,sinayoko2011}
is, in the present language, the practical trace of the cocycle
\eqref{eq:cocycle} and of the transport defects computed in
\S\S\ref{sec:holonomy} and~\ref{sec:stat}.  The complementary, fixed-operator part of the story is the
theory of non-radiating sources \citep{jessel1973,kempton1976,
bleistein1977,porterdevaney1982,marengo2000}: \citet{musafir2013},
following \citet{doak1988}, classifies the fibre of equivalent sources
over a single analogy, and the term-shifting freedom among Lilley source
forms catalogued by \citet{musafir2007} is precisely a path within that
fibre.  The present paper supplies the transversal geometry---the maps
\emph{between} fibres---and shows where its obstructions lie.  The
observability limits of holographic and inverse methods
\citep{maynard1985} are consistent with this division: only the
sonic-sphere trace is invariant, exactly as \eqref{eq:pipelines}
transports it.

\subsection{Open problems}

Five problems are open, in what we judge to be increasing order of
depth.
(1)~\emph{Finite refinement.}  Prove (or refute)
conjecture~\ref{conj:evasion} on horizon evasion by $n$-step paths.
(2)~\emph{Tensor sources.}  Extend theorem~\ref{thm:pipelines} and
lemma~\ref{lem:anisofar} to tensor source models, whose polarisation
transports nontrivially on the sonic ellipsoid; the scalar
anisotropic-endpoint case is closed by \S\ref{ssec:anisofar}.
(3)~\emph{Torsion.}  The canonical connection carries torsion in mixed
boost--anisotropy planes (theorem~\ref{thm:curvature}); its physical
meaning is unknown.
(4)~\emph{Variable coefficients.}  Genuinely sheared base flows
\citep{pridmorebrown1958,dowling1978,goldstein2003} replace the matrix
group by Fourier integral operators; the uniform theory developed here is
the principal-symbol shadow of that theory, and the reorganisations of
the source around vorticity and dilatation
\citep{ribner1962,powell1964,howe1975,mohring1978,doak1989,crow1970}
additionally change the variable map $\varphi$, adding a layer to the
fibre structure.
(5)~\emph{Statistical invariants.}  The transport of second-order
statistics in lemma~\ref{lem:csd} suggests that the invariant content of
a statistical source model under re-gauging is microlocal---an
H-measure or microlocal defect measure \citep{tartar1990,gerard1991}
attached to the analogy class rather than to any single analogy.
Developing that programme, and connecting it to the practice reviewed by \citet{goldstein1984},
\citet{tam1998} and \citet{jordancolonius2013}, is in our view the
natural continuation.

\section{Acknoledgement}
The author thanks Lorna~J.~Ayton (Department of
Applied Mathematics and Theoretical Physics, University of Cambridge) for
valuable discussions.

\section{Declaration of the use of AI tools}The author used a large
language model, Apertus-8B-Instruct-2509 \citep{swissai2025apertus} (Swiss AI
Initiative; hosted on the Blablador service of the J\"ulich Supercomputing
Centre, Helmholtz Association; \url{https://huggingface.co/swiss-ai/Apertus-8B-Instruct-2509};
accessed July 2026), for English-language grammar and phrasing checks of the
author's own draft text. The tool was not used to generate scientific content,
to produce or analyse data, or to generate figures; all derivations, numerical
results, figures and interpretations are the author's own. Accountability for
the final text rests with the author, who has checked the manuscript, and in
particular the references, for any unintended consequences of this use.

\medskip
\noindent\textbf{Data availability.} All symbolic and numerical
verifications reported in this paper are reproduced by the supplementary
scripts \texttt{supplementary\_verification.py} and
\texttt{supplementary\_verification\_farfield.py} (SymPy/NumPy), which
accompany the manuscript.

\appendix

\section{Verification record}\label{app:verification}

Every identity asserted in the paper was verified independently of its
derivation: symbolic claims by exact rational--radical arithmetic in
SymPy, numerical claims in double-precision NumPy through the full
transport pipeline.  The supplementary scripts print each residual; all
residuals below are identically zero (symbolic) or at machine precision
(numerical).

\begin{itemize}
\item[V1] Kinematic cocycle identity \eqref{eq:kinematic}, three space
dimensions, arbitrary fields.
\item[V2] Fixed-frequency reduction \eqref{eq:PGL}, including the
uniqueness of the phase $\alpha$ (lemma~\ref{lem:PGL}).
\item[V3] (a) Operator congruence \eqref{eq:wreduces} for the elementary
map; (b--d) the perpendicular endpoint data
\eqref{eq:endpoint-perp}--\eqref{eq:A3}: $\det C_3-c^4=0$,
$A_3=\operatorname{diag}(c^2-u_1^2,c^2-u_2^2)$, $q_{00}=1$
(theorem~\ref{thm:nonclosure}).
\item[V4] (a) $\bigl(W^{-1}\bigr)_{00}=\gamma_1\gamma_2$
(lemma~\ref{lem:firstrow}); (b) the closed form \eqref{eq:rho-exact} for
general step directions; (c) the identity
$q_{00}=\gamma_1^2\gamma_2^2(1-|\bw|^2/c^2)$; (d) the exact coincidence
of $\bw$ with the relativistic composition for perpendicular steps
(corollary~\ref{cor:einstein}).
\item[V5] (a) The collinear conformal factor in \eqref{eq:collinear-v3};
(b) the exact factorised forms \eqref{eq:collinear-A}; (c) the symbolic
solution of the equal-step horizon, returning
$u=c(\sqrt5-1)/2$ \eqref{eq:golden}.
\item[V6] (a) The boost generator \eqref{eq:bgen} along three
directions; (b--d) the brackets
\eqref{eq:brackets-bm}--\eqref{eq:brackets-aa}
(theorem~\ref{thm:curvature}).
\item[V7] (a) Boost-loop holonomy: $\rho$ from the full numerical
pipeline agrees with \eqref{eq:rho-exact} to 15 significant digits and
$|D-\rho I|_{\max}=2.2\times10^{-16}$; (b) anisotropy-loop holonomy:
$\rho=1$, $\theta\to\varepsilon_1\varepsilon_2/4$
(proposition~\ref{prop:bch}); (c) far-field directivity of a quadrupole
model rotates rigidly by $\theta$ under loop transport, to the angular
grid resolution (theorem~\ref{thm:pipelines};
table~\ref{tab:verification}).
\item[V8] The far-field lemma (lemma~\ref{lem:anisofar}): (a) the
canonical reduction $w_3\,Q(\bv,C)\,w_3\T=\Qstar$ in general symbolic
form; (b) the intrinsic identifications for generic three-dimensional
media: $\bk_*\in E_\omega$,
$\omega-\bv\!\cdot\!\bk_*=(\omega/p)\mu$,
$\bk_*\T C\bk_*=(\omega-\bv\!\cdot\!\bk_*)^2$,
$\bv_g(\bk_*)=\boldsymbol{\nu}/\mu$, and the radial phase identity in
\eqref{eq:kstar}; (c) end-to-end numerical validation of
\eqref{eq:anisofar} against a direct limiting-absorption FFT solution of
the anisotropic convected Helmholtz equation ($M\approx0.32$, fully
anisotropic $C$), agreeing to $0.35\%$ in amplitude and
$0.033\,$rad in phase across downstream, transverse and oblique
directions at $k_*r\approx28$--$48$, the upstream direction
($k_*r\approx55$) converging at the rate of the periodic-image damping
as the absorption parameter is increased (documented in the script); and
the group-ray characterisation confirmed by direct search on the
ellipsoid, matching $\bk_*$ to $10^{-6}$.  Script
\texttt{supplementary\_verification\_farfield.py}.
\end{itemize}

\bibliographystyle{plainnat}
\bibliography{references}

\end{document}